\documentclass[]{interact}
\usepackage{epstopdf}
\usepackage[caption=false]{subfig}

\usepackage[longnamesfirst,sort]{natbib}
\bibpunct[, ]{(}{)}{;}{a}{,}{,}
\usepackage{natbib}

\usepackage{amsmath}
\usepackage{amssymb}
\usepackage{algorithmic}
\usepackage{algorithm}
\usepackage{caption}
\usepackage{mathtools}
\usepackage{changes}
\usepackage{multirow}
\usepackage{verbatim}
\usepackage{mathrsfs}
\usepackage{graphicx}
\usepackage{amsmath,amsthm,bm,mathrsfs}

\theoremstyle{plain}
\newtheorem{theorem}{Theorem}[section]

\newtheorem{proposition}[theorem]{Proposition}

\theoremstyle{definition}

\theoremstyle{remark}
\newtheorem{remark}{Remark}

\begin{document}

\title{Frequency-weighted $\mathcal{H}_2$-optimal model order reduction via oblique projection}

\author{
\name{Umair~Zulfiqar\textsuperscript{a}\thanks{CONTACT Umair~Zulfiqar. Email: umair.zulfiqar@research.uwa.edu.au}, Victor Sreeram\textsuperscript{a}, Mian~Ilyas~Ahmad\textsuperscript{b}, and Xin~Du\textsuperscript{c,d,e}}
\affil{\textsuperscript{a}School of Electrical, Electronics and Computer Engineering, The University of Western Australia (UWA), Perth, Australia; \textsuperscript{b}Research Centre for Modelling and Simulation, National University of Sciences and Technology (NUST), Islamabad, Pakistan; \textsuperscript{c}School of Mechatronic Engineering and Automation, and Shanghai Key Laboratory of Power Station Automation Technology, Shanghai University, Shanghai, China; \textsuperscript{d}Key Laboratory of Knowledge Automation for Industrial Processes, Ministry of Education, Beijing, China; \textsuperscript{e}Key Laboratory of Modern Power System Simulation and Control \& Renewable Energy Technology, Ministry of Education (Northeast Electric Power University), Jilin, China}
}

\maketitle

\begin{abstract}
In projection-based model order reduction, a reduced-order approximation of the original full-order system is obtained by projecting it onto a reduced subspace that contains its dominant characteristics. The problem of frequency-weighted $\mathcal{H}_2$-optimal model order reduction is to construct a local optimum in terms of the $\mathcal{H}_2$-norm of the weighted error transfer function. In this paper, a projection-based model order reduction algorithm is proposed that constructs a reduced-order model, which nearly satisfies the first-order optimality conditions for the frequency-weighted $\mathcal{H}_2$-optimal model order reduction problem. It is shown that as the order of the reduced model is increased, the deviation in the satisfaction of the optimality conditions reduces further. Numerical methods are also discussed that improve the computational efficiency of the proposed algorithm. Four numerical examples are presented to demonstrate the efficacy of the proposed algorithm.
\end{abstract}


\section{Introduction}
The complexity of the modern-day dynamic systems has been growing rapidly with each passing day. The direct simulation of the high-order mathematical models that describe large-scale dynamic systems requires a huge amount of computational resources, which are limited due to high economic cost of the memory resources. To address this issue, model order reduction (MOR) algorithms are used to obtain reduced-order models (ROMs) that are computationally cheaper to simulate, and they closely mimic the original high-order models. The reduced models can then be used as surrogates in the design and analysis with tolerable approximation error \citep{antoulas2005approximation,benner2005dimension,benner2017model}.

Projection-based MOR is a large family of algorithms wherein the original high-order model is projected onto a reduced subspace that contains its dominant characteristics. The sense of dominance determines the specific type of MOR procedure. Most of the projection-based MOR methods require solutions of some Lyapunov or Sylvester equations to construct the required ROM \citep{antoulas2005approximation}. During the last two decades, several computationally efficient low-rank methods for the solution of Lyapunov and Sylvester equations have been proposed, cf. \citep{ahmad2010krylov,gugercin2003modified,li2002low,penzl1999cyclic}. By using these methods for solving large-scale linear matrix equations, a ROM of the original large-scale model can be constructed within the admissible time for most of the projection-based MOR algorithms.

Balanced truncation (BT) \citep{moore1981principal} is one of the most important projection-based MOR methods, which is famous for its stability preservation, high fidelity, and \textit{apriori} error bound expression \citep{enns1984model}. In some situations like reduced-order controller design, it is required that the MOR procedure has a low frequency-weighted approximation error. This necessitates the inclusion of frequency-weights in the MOR algorithm. In \citep{enns1984model}, BT is generalized to incorporate frequency weights in the approximation criterion, which leads to frequency-weighted BT (FWBT). Several modifications and extensions to FWBT are reported in the literature to ensure additional properties like stability \citep{wang1999new} and passivity \citep{zulfiqar2017passivity}; see \citep{ghafoor2008survey,obinata2012model} for a detailed survey.

Another important class of projection-based MOR techniques is the Krylov subspace-based methods wherein the full-order system is projected onto a low-dimensional subspace spanned by the columns of a matrix constructed so that the projected reduced system achieves moment matching, i.e., the ROM matches some coefficients of the series expansion of the original transfer function at some selected frequency points \citep{beattie2014model}. Among these methods is the famous iterative rational Krylov algorithm (IRKA) \citep{gugercin2008h_2,van2008h2}, which constructs a local optimum for the $\mathcal{H}_2$-optimal MOR problem, i.e., the best among all the ROMs with the same modal configuration and size in minimizing the $\mathcal{H}_2$-norm of the error transfer function. Unlike the BT method, IRKA \citep{gugercin2008h_2,van2008h2} does not require the solutions of large-scale Lyapunov equations. Thus it is computationally efficient and can handle large-scale systems. Some other projection-based algorithms for the $\mathcal{H}_2$-optimal MOR problem include but are not limited to \citep{ahmad20100,ibrir2018projection,yan1999approximate}. IRKA is heuristically generalized to the frequency-weighted scenario in \citep{anic2013interpolatory} and \citep{zulfiqar2018weighted}. The algorithms proposed in \citep{anic2013interpolatory} and \citep{zulfiqar2018weighted} ensure less $\mathcal{H}_2$-norm of the weighted error transfer function; however, they do not seek to construct local optimum for the frequency-weighted $\mathcal{H}_2$-optimal MOR problem.

In \citep{halevi1990frequency}, the first-order optimality conditions for the single-sided case of frequency-weighted $\mathcal{H}_2$-optimal MOR are derived, and an algorithm based on Lyapunov and Riccati equations is proposed to satisfy these conditions. This algorithm is numerically tractable only for small-scale systems. The optimality conditions derived in \citep{halevi1990frequency} are shown equivalent to the tangential interpolation conditions in \citep{breiten2015near}, and a Krylov subspace-based iterative algorithm is proposed, which nearly satisfies these conditions. In \citep{zulfiqar2019frequency}, an iteration-free Krylov subspace-based algorithm is presented, which exactly satisfies a subset of the optimality conditions while guaranteeing the stability of the ROM at the same time.

The first-order optimality conditions for the double-sided case of the frequency-weighted $\mathcal{H}_2$-optimal MOR are derived in \citep{diab2000optimal,petersson2013nonlinear,yan1997convergent}. The algorithms presented to generate the local optimum in \citep{diab2000optimal,huang2001new,li1999h,petersson2013nonlinear,spanos1990optimal,yan1997convergent} are not feasible for large-scale systems due to high computational cost associated with nonlinear optimization. In \citep{zulfiqar2019frequency}, the frequency-weighted iterative tangential interpolation algorithm (FWITIA) is proposed that ensures less frequency-weighted $\mathcal{H}_2$-norm of the error transfer function in the double-side case. However, its connection with the optimality conditions derived in \citep{petersson2013nonlinear} is not investigated, which has motivated the work in this paper.

In this paper, we consider the double-sided case of the frequency-weighted $\mathcal{H}_2$-optimal MOR problem within the projection framework. The main motivation for seeking a projection-based solution is to avoid nonlinear optimization and benefit from the efficient Sylvester equation solver, particularly formulated for the projection-based $\mathcal{H}_2$-optimal MOR algorithms, cf. \citep{bennersparse}. We show that the exact satisfaction of the optimality conditions is inherently not possible within the projection framework. However, the optimality conditions can be nearly satisfied, and the deviation in the satisfaction of the optimality condition decays as the order of ROM grows. The conditions for exact satisfaction of the optimality conditions are also discussed. In addition, a projection-based iterative algorithm is proposed that solves the Sylvester equations in each iteration to construct the required ROM. Upon convergence, the ROM nearly satisfies the first-order optimality conditions for the double-sided case of the frequency-weighted $\mathcal{H}_2$-optimal MOR problem. Moreover, it is shown that FWITIA also seeks to satisfy the optimality conditions derived in \citep{petersson2013nonlinear}. The efficacy of the proposed algorithm is highlighted by considering one illustrative and three benchmark numerical examples.

The remainder of this paper is organized as follows. The problem under consideration is introduced in Section \ref{sec1}. In addition, FWITIA is briefly reviewed. The main contribution of the paper is covered in Sections \ref{sec2} and \ref{sec3}. The theoretical results proposed in Sections \ref{sec2} and \ref{sec3} are numerically verified in Section \ref{sec4}. The paper is concluded in Section \ref{sec5}.
\section{Preliminaries}\label{sec1}
In this section, the frequency-weighted $\mathcal{H}_2$-optimal MOR problem is introduced, and FWITIA is briefly reviewed. The important mathematical notations used throughout the text are given in Table \ref{tab0}.
\begin{table}[!h]
\centering
\caption{Mathematical Notations}\label{tab0}
\begin{tabular}{|c|c|}
\hline
Notation & Meaning \\ \hline
$\begin{bmatrix}\cdot\end{bmatrix}^*$  & Hermitian\\
$tr(\cdot)$  & Trace\\
$Ran(\cdot)$ & Range\\
$orth(\cdot)$& Orthogonal basis\\
$\bot$  & Orthogonal\\
$\supset$& Superset\\
$\underset {i=1,\cdots,r}{span}\{\cdot\}$ & Span of the set of $r$ vectors\\
\hline
\end{tabular}
\end{table}
\subsection{Problem Setting}
Let us denote an $n^{th}$-order stable linear time-invariant system with $m$ inputs and $p$ outputs as $H(s)$, which is represented as
\begin{align}
 H(s)&=C(sI-A)^{-1}B+D,\label{e1}
 \end{align} where $A\in\mathbb{R}^{n\times n}$, $B\in\mathbb{R}^{n\times m}$, $C\in\mathbb{R}^{p\times n}$, and $D\in\mathbb{R}^{p\times m}$. In a large-scale setting, the order $n$ of (\ref{e1}) is high, the matrices $(A,B,C)$ are sparse, $m\ll n$, and $p\ll n$.

 Let us denote the $r^{th}$-order approximation of $H(s)$ as $\tilde{H}(s)$, which can be written as
 \begin{align}
 \tilde{H}(s)&=\tilde{C}(sI-\tilde{A})^{-1}\tilde{B}+D,\nonumber
 \end{align}
 where $\tilde{A}\in\mathbb{R}^{r\times r}$, $\tilde{B}\in\mathbb{R}^{r\times m}$, $\tilde{C}\in\mathbb{R}^{p\times r}$, and $r\ll n$.

 In projection-based MOR, the state-space matrices $\tilde{A}$, $\tilde{B}$, and $\tilde{C}$ are obtained as
 \begin{align}
 \tilde{A}&=\tilde{W}^TA\tilde{V},&\tilde{B}&=\tilde{W}^TB,&\tilde{C}&=C\tilde{V},\label{E2}
 \end{align}
 where $\tilde{V}\in\mathbb{R}^{n\times r}$, $\tilde{W}\in\mathbb{R}^{n\times r}$ and $\tilde{W}^T\tilde{V}=I$. The columns of $\tilde{V}$ span $r$-dimensional subspace along the kernel of $\tilde{W}^T$, and $\Pi=\tilde{V}\tilde{W}^T\in\mathbb{R}^{n\times n}$ is an oblique projection onto that subspace.

 Let us denote the error transfer function as $E(s)$, which can be written as
 \begin{align}
 E(s)&=H(s)-\tilde{H}(s)=C_e(sI-A_e)^{-1}B_e\nonumber
 \end{align} in which
 \begin{align}
 A_e&=\begin{bmatrix}A&0\\0&\tilde{A}\end{bmatrix},&B_e&=\begin{bmatrix}B\\\tilde{B}\end{bmatrix},&C_e&=\begin{bmatrix}C&-\tilde{C}\end{bmatrix}.\nonumber
 \end{align}

Let us denote the input and output weights as $W_i(s)$ and $W_o(s)$, respectively, which have the following transfer functions
 \begin{align}
  W_i(s)&=C_i(sI-A_i)^{-1}B_i+D_i,\nonumber\\
  W_o(s)&=C_o(sI-A_o)^{-1}B_o+D_o,\nonumber
  \end{align} where $A_i\in\mathbb{R}^{n_i\times n_i}$, $B_i\in\mathbb{R}^{n_i\times m}$, $C_i\in\mathbb{R}^{m\times n_i}$, $D_i\in\mathbb{R}^{m\times m}$, $A_o\in\mathbb{R}^{n_o\times n_o}$, $B_o\in\mathbb{R}^{n_o\times p}$, $C_o\in\mathbb{R}^{p\times n_o}$, and $D_o\in\mathbb{R}^{p\times p}$. Also, assume that $W_i(s)$ and $W_o(s)$ are stable.

 Further, let us denote the weighted error transfer function as $E_w(s)$, which has the following representation
 \begin{align}
 E_w(s)&=W_o(s)E(s)W_i(s)=C_w(sI-A_w)^{-1}B_w,\nonumber
 \end{align} where
 \begin{align}
 A_w&=\begin{bmatrix}A&0&BC_i&0\\0&\tilde{A}&\tilde{B}C_i&0\\0&0&A_i&0\\B_oC&-B_o\tilde{C}&0&A_o\end{bmatrix},&B_w&=\begin{bmatrix}BD_i\\\tilde{B}D_i\\B_i\\0\end{bmatrix},\nonumber\\
 C_w&=\begin{bmatrix}D_oC&-D_o\tilde{C}&0&C_o\end{bmatrix}.&&\label{e2}
 \end{align}

Let us denote the controllability and observability gramians of the realization $(A_w,B_w,C_w)$ as $P_w$ and $Q_w$, respectively, which solve the following Lyapunov equations
 \begin{align}
 A_wP_w+P_wA_w^T+B_wB_w^T&=0,\label{e3}\\
 A_w^TQ_w+Q_wA_w+C_w^TC_w&=0.\label{e4}
 \end{align} $P_w$ and $Q_w$ can be partitioned according to the structure of the realization in (\ref{e2}) by defining
 \begin{align}
 P_w&=\begin{bmatrix}P&P_{12}&P_{13}&P_{14}\\P_{12}^T&\tilde{P}&P_{23}&P_{24}\\P_{13}^T&P_{23}^T&P_i&P_{34}\\P_{14}^T&P_{24}^T&P_{34}^T&P_o\end{bmatrix},&Q_w&=\begin{bmatrix}Q&Q_{12}&Q_{13}&Q_{14}\\Q_{12}^T&\tilde{Q}&Q_{23}&Q_{24}\\Q_{13}^T&Q_{23}^T&Q_i&Q_{34}\\Q_{14}^T&Q_{24}^T&Q_{34}^T&Q_o\end{bmatrix}.\nonumber
 \end{align}
 The $\mathcal{H}_2$-norm of $E_w(s)$ is the energy of the its impulse response and is related to the controllability and observability gramians of its state-space realization as
 \begin{align}
|| E_w(s)||_{\mathcal{H}_2}^2&=tr(C_wP_wC_w^T)=tr\big(D_oCPC^TD_o^T+D_o\tilde{C}\tilde{P}\tilde{C}^TD_o^T-2D_oCP_{12}\tilde{C}^TD_o^T\nonumber\\
&\hspace*{4cm}+C_oP_oC_o^T+2D_oCP_{14}C_o^T-2D_o\tilde{C}P_{24}C_o^T\big)\nonumber\\
&=tr(B_w^TQ_wB_w)=tr\big(D_i^TB^TQBD_i+D_i^T\tilde{B}^T\tilde{Q}\tilde{B}D_i+2D_i^TB^TQ_{12}\tilde{B}D_i\nonumber\\
&\hspace*{4cm}+B_i^TQ_iB_i+2D_i^TB^TQ_{13}B_i+2D_i^T\tilde{B}^TQ_{23}B_i\big).\nonumber
 \end{align}
The local optimum $\tilde{H}(s)$ of $||E_w(s)||_{\mathcal{H}_2}^2$ satisfies the following first-order optimality conditions, cf. \citep{petersson2013nonlinear},
 \begin{align}
 \frac{\partial}{\partial \tilde{A}}||E_w(s)||_{\mathcal{H}_2}^2&=0&\Rightarrow  &&\bar{X}+X=0,\label{e5}\\
\frac{\partial}{\partial \tilde{B}}||E_w(s)||_{\mathcal{H}_2}^2&=0&\Rightarrow&&\bar{Y}D_iD_i^T+Y=0,\label{e6}\\
\frac{\partial}{\partial \tilde{C}}||E_w(s)||_{\mathcal{H}_2}^2&=0&\Rightarrow &&D_o^TD_o\bar{Z}+Z=0\label{e7}
 \end{align} where
 \begin{align}
 \bar{X}&=Q_{12}^TP_{12}+\tilde{Q}\tilde{P},&\bar{Y}&=Q_{12}^TB+\tilde{Q}\tilde{B},\nonumber\\
\bar{Z}&=CP_{12}-\tilde{C}\tilde{P},&X&=Q_{23}P_{23}^T+Q_{24}P_{24}^T,\nonumber\\
 Y&=\big(Q_{12}^TP_{13}+\tilde{Q}P_{23}+Q_{23}P_i+Q_{24}P_{34}^T\big)C_i^T+Q_{23}B_iD_i^T,\nonumber\\
 Z&=-B_o^T\big(Q_{14}^TP_{12}+Q_{24}^T\tilde{P}+Q_{34}^TP_{23}^T+Q_oP_{24}^T\big)+D_o^TC_oP_{24}^T.\nonumber
 \end{align}
The problem under consideration, i.e., the frequency-weighted $\mathcal{H}_2$-MOR problem, is the following. Given the $n^{th}$-order original system $H(s)$, the input weight $W_i(s)$, and the output weight $W_o(s)$, we need to construct an $r^{th}$-order reduced model $\tilde{H}(s)$, which ensures that the $\mathcal{H}_2$-norm of $E_w(s)$ is small, i.e.,
 \begin{align}
 \underset{\substack{\tilde{H}(s)\\\textnormal{order}=r}}{\text{min}}||E_w(s)||_{\mathcal{H}_2}.\nonumber\end{align}
\subsection{Frequency-weighted Tangential Interpolation}
Here, the goal is to construct a ROM that tends to satisfy the following conditions
\begin{align}
\frac{\partial}{\partial \tilde{B}}||W_o(s)E(s)||_{\mathcal{H}_2}^2&=2(Q_{12}^TB+\tilde{Q}\tilde{B})=0,\label{e8}\\
\frac{\partial}{\partial \tilde{C}}||E(s)W_i(s)||_{\mathcal{H}_2}^2&=2(CP_{12}-\tilde{C}\tilde{P})=0.\label{e9}
\end{align} Suppose $H(s)$ and $\tilde{H}(s)$ have simple poles. Also, let $\tilde{H}(s)$ has the following pole-residue form
\begin{align}
\tilde{H}(s)&=\sum_{i=1}^{r}\frac{\tilde{l}_i\tilde{r}_i^T}{s-\tilde{\lambda}_i}+D.\nonumber
\end{align}
Now define $F\big[H(s)\big]=C_f(sI-A_f)^{-1}B_f$ and $G\big[H(s)\big]=C_g(sI-A_g)^{-1}B_g$ wherein
\begin{align}
A_f&=\begin{bmatrix}A&BC_i\\0&A_i\end{bmatrix},&A_g&=\begin{bmatrix}A&0\\B_oC&A_o\end{bmatrix},\label{e10}\\
B_f&=\begin{bmatrix}P_{13}C_i^T+BD_iD_i^T\\P_iC_i^T+B_iD_i^T\end{bmatrix},& B_g&=\begin{bmatrix}B\\B_oD\end{bmatrix},\label{e11}\\
C_f&=\begin{bmatrix}C^T\\C_i^TD_i^T\end{bmatrix}^T,& C_g&=\begin{bmatrix}Q_{14}B_o+C^TD_o^TD_o\\ Q_oB_o+C_o^TD_o\end{bmatrix}^T.\label{e12}
\end{align}
The conditions (\ref{e8}) and (\ref{e9}) are satisfied when the following tangential interpolation conditions are satisfied, cf. \citep{zulfiqar2019frequency},
\begin{align}
F\big[H(-\tilde{\lambda}_i)\big]\tilde{r}_i&=\tilde{F}\big[\tilde{H}(-\tilde{\lambda}_i)\big]\tilde{r}_i,\label{e13}\\ \tilde{l}_i^TG\big[H(-\tilde{\lambda}_i)\big]&=\tilde{l}_i^T\tilde{G}\big[\tilde{H}(-\tilde{\lambda}_i)\big].\label{e14}
\end{align}
The poles $\tilde{\lambda}_i$ and residues $(\tilde{l}_i,\tilde{r}_i)$ of $\tilde{H}(s)$ are not known \textit{apriori}. Thus the interpolation points and tangential directions are initialized arbitrarily, and after every iteration, the interpolation points $\sigma_i$ are updated as $-\tilde{\lambda}_i$, and the tangential directions $(c_i,b_i)$ are updated as the residues $(\tilde{l}_i,\tilde{r}_i)$. The rational Krylov subspaces that seek to satisfy the tangential interpolation conditions (\ref{e13}) and (\ref{e14}) are computed as
\begin{align}
Ran\begin{bmatrix}V_a\\V_b\end{bmatrix}&=\underset {i=1,\cdots,r}{span}\{(\sigma_i I-A_f)^{-1}B_{f} b_i\},\nonumber\\
Ran\begin{bmatrix}W_a\\W_b\end{bmatrix}&=\underset {i=1,\cdots,r}{span}\{(\sigma_i I-A_g^T)^{-1}C_g^T c_i\}.\nonumber
\end{align} $\tilde{V}$ and $\tilde{W}$ are set as $Ran(\tilde{V})\supset Ran(V_a)$ and $Ran(\tilde{W})\supset Ran(W_a)$, $\tilde{V}=orth(\tilde{V})$, $\tilde{W}=orth(\tilde{W})$, $\tilde{W}=\tilde{W}(\tilde{V}^T\tilde{W})^{-1}$. The algorithm is stopped when the relative change in $\tilde{\lambda}_i$ stagnates, and a ROM $\tilde{H}(s)$ is identified that tends to satisfy the interpolation conditions (\ref{e13}) and (\ref{e14}).
\section{Frequency-weighted $\mathcal{H}_2$-optimal MOR}\label{sec2}
In this section, it is shown that the optimality conditions (\ref{e5})-(\ref{e7}) cannot be (inherently) satisfied exactly within the projection framework. However, the conditions $\bar{X}=0$, $\bar{Y}=0$, and $\bar{Z}=0$ can be achieved with the projection approach. The conditions for exactly satisfying the optimality conditions (\ref{e5})-(\ref{e7}) are also discussed. Further, it is shown that as $r$ increases, the deviation in the satisfaction of the optimality conditions (\ref{e5})-(\ref{e7}) reduces as long as $\bar{X}=0$, $\bar{Y}=0$, and $\bar{Z}=0$.
\subsection{Limitation of Projection Framework}
Let us assume at the moment for simplicity that $D_iD_i^T$, $D_o^TD_o$, $\tilde{P}$, and $\tilde{Q}$ are invertible. Then it can be noted from the optimality conditions (\ref{e6}) and (\ref{e7}) that the optimal choices of $\tilde{B}$ and $\tilde{C}$ should satisfy the following
\begin{align}
\tilde{B}&=-\tilde{Q}^{-1}Q_{12}^TB-\tilde{Q}^{-1}Y(D_iD_i^T)^{-1},\nonumber\\
\tilde{C}&=CP_{12}\tilde{P}^{-1}+(D_o^TD_o)^{-1}Z\tilde{P}^{-1}.\nonumber
\end{align}
Since $\tilde{B}$ and $\tilde{C}$ are computed as $\tilde{W}^TB$ and $C\tilde{V}$, respectively, in the projection framework, the matrices $\tilde{V}$ and $\tilde{W}$ should be selected as $\tilde{V}=P_{12}\tilde{P}^{-1}$ and $\tilde{W}=-Q_{12}\tilde{Q}^{-1}$ (note that $\tilde{P}$ and $\tilde{Q}$ are symmetric). Moreover, this selection ensures that $\bar{Y}=0$ and $\bar{Z}=0$ without assuming that $D_iD_i^T$ and $D_o^TD_o$ are invertible. Since $P_{12}$, $Q_{12}$, $\tilde{P}$, and $\tilde{Q}$ depend on the unknown $(\tilde{A},\tilde{B},\tilde{C})$, the problem is nonconvex. Nevertheless, if such a solution is found within the projection framework, $\bar{X}=0$ due to the oblique projection condition $\tilde{W}^T\tilde{V}=I$. The deviations in the satisfaction of optimality conditions (\ref{e5})-(\ref{e7}) are then quantified by $X$, $Y$, and $Z$. If $X=0$, $Y=0$, and $Z=0$, the problem can be solved within the projection framework by finding the reduction matrices that ensure $\tilde{V}=P_{12}\tilde{P}^{-1}$, $\tilde{W}=-Q_{12}\tilde{Q}^{-1}$, and $\tilde{W}^T\tilde{V}=I$. The reduction matrices $\tilde{V}$ and $\tilde{W}$ have no influence on $P_{13}$, $P_i$, $Q_{14}$, and $Q_o$. There seems no straightforward way to influence $P_{23}$, $P_{24}$, $P_{34}$, $Q_{23}$, $Q_{24}$, and $Q_{34}$ using $\tilde{V}$ and $\tilde{W}$ so that $X$, $Y$, and $Z$ become zeros. Further, it is shown in \citep{hurak2001discussion,sreeram2002properties,sreeram2012improved} that the nonzero cross-terms $P_{13}$, $P_{23}$, $Q_{14}$, and $Q_{24}$ are inherent to the frequency-weighted MOR problem, and the effect of frequency-weights vanishes when these are zeros. Therefore, at best, we can ensure $\bar{X}=0$, $\bar{Y}=0$, and $\bar{Z}=0$ within the projection framework.

FWITIA is not motivated by the optimality conditions (\ref{e5})-(\ref{e7}). Instead, it follows the system theory perspective that ensuring small $||E(s)W_i(s)||_{\mathcal{H}_2}$ and $||W_o(s)E(s)||_{\mathcal{H}_2}$ generally ensures that $||E_w(s)||_{\mathcal{H}_2}$ is also small. Note that $\mathcal{H}_2$-norm, unlike $\mathcal{H}_\infty$-norm, does not enjoy the submultiplicative property, and hence,
\begin{align}
||E_w(s)||_{\mathcal{H}_2}&\leq||E(s)W_i(s)||_{\mathcal{H}_2}+||W_o(s)E(s)||_{\mathcal{H}_2}\nonumber
\end{align} does not hold in general. However, it can be shown that $\bar{Y}=0$ and $\bar{Z}=0$ is equivalent to ensuring that the gradients of the additive components of $||E_w(s)||_{\mathcal{H}_2}^2$ with respect to $\tilde{B}$ and $\tilde{C}$, respectively, become zero. This is established in Proposition \ref{l1}.
\begin{proposition}\label{l1}
Let us split $||E_w(s)||_{\mathcal{H}_2}^2$ into its additive components as $||E_w(s)||_{\mathcal{H}_2}^2=\mathcal{J}_1+\mathcal{J}_2=\mathcal{J}_3+\mathcal{J}_4$ where
\begin{align}
\mathcal{J}_1&=tr(D_oCPC^TD_o^T+D_o\tilde{C}\tilde{P}\tilde{C}^TD_o^T-2D_oCP_{12}\tilde{C}^TD_o^T),\nonumber\\
\mathcal{J}_2&=tr(C_oP_oC_o^T+2D_oCP_{14}C_o^T-2D_o\tilde{C}P_{24}C_o^T),\nonumber\\
\mathcal{J}_3&=tr(D_i^TB^TQBD_i+D_i^T\tilde{B}^T\tilde{Q}\tilde{B}D_i+2D_i^TB^TQ_{12}\tilde{B}D_i),\nonumber\\
\mathcal{J}_4&=tr(B_i^TQ_iB_i+2D_i^TB^TQ_{13}B_i+2D_i^T\tilde{B}^TQ_{23}B_i).\nonumber
\end{align}
Then $\frac{\partial}{\partial\tilde{C}}\mathcal{J}_1=0$ and $\frac{\partial}{\partial\tilde{B}}\mathcal{J}_3=0$ when $\bar{Z}=0$ and $\bar{Y}=0$, respectively.
\end{proposition}
\begin{proof}
Let us denote the first-order derivative of $\mathcal{J}_3$ with respect to $\tilde{B}$ as $\Delta_{\mathcal{J}_3}^{\tilde{B}}$ and the differential of $\tilde{B}$ as $\Delta_{\tilde{B}}$. Then
\begin{align}
\Delta_{\mathcal{J}_3}^{\tilde{B}}&=tr(2D_i^T\Delta_{\tilde{B}}^T\tilde{Q}\tilde{B}D_i+2D_i^TB^TQ_{12}\Delta_{\tilde{B}}D_i)\nonumber\\
&=tr\Big(\big(2D_iD_i^T\tilde{B}\tilde{Q}+2D_iD_i^TB^TQ_{12}\big)\Delta_{\tilde{B}}\Big).\nonumber
\end{align}
Since $\Delta_{\mathcal{J}_3}^{\tilde{B}}=tr\Big(\big(\frac{\partial}{\partial\tilde{B}}\mathcal{J}_3\big)^T\Delta_{\tilde{B}}\Big)$, $\frac{\partial}{\partial\tilde{B}}\mathcal{J}_3=2\bar{Y}D_iD_i^T$. Thus when $\bar{Y}=0$, $\frac{\partial}{\partial\tilde{B}}\mathcal{J}_3=0$.

Now, let us denote the first-order derivative of $\mathcal{J}_1$ with respect to $\tilde{C}$ as $\Delta_{\mathcal{J}_1}^{\tilde{C}}$ and the differential of $\tilde{C}$ as $\Delta_{\tilde{C}}$. Then
\begin{align}
\Delta_{\mathcal{J}_1}^{\tilde{C}}&=tr(2D_o\Delta_{\tilde{C}}\tilde{P}\tilde{C}^TD_o^T-2D_oCP_{12}\Delta_{\tilde{C}}^TD_o^T)\nonumber\\
&=tr\Big(\big(2\tilde{P}\tilde{C}^TD_o^TD_o-2P_{12}^TC^TD_o^TD_o\big)\Delta_{\tilde{C}}\Big).\nonumber
\end{align}
Since $\Delta_{\mathcal{J}_1}^{\tilde{C}}=tr\Big(\big(\frac{\partial}{\partial\tilde{C}}\mathcal{J}_1\big)^T\Delta_{\tilde{C}}\Big)$, $\frac{\partial}{\partial\tilde{C}}\mathcal{J}_1=2D_o^TD_o\bar{Z}$. Thus when $\bar{Z}=0$, $\frac{\partial}{\partial\tilde{C}}\mathcal{J}_1=0$. This completes the proof.
\end{proof}
  Although this was not recognized in \citep{zulfiqar2019frequency}, it is evident now that FWITIA is not completely heuristic in terms of seeking to ensure that $||E_w(s)||_{\mathcal{H}_2}^2$ is small. Note that FWITIA seeks to ensure that $\bar{Y}=0$ and $\bar{Z}=0$ by satisfying the tangential interpolation conditions (\ref{e13}) and (\ref{e14}). However, (\ref{e13}) and (\ref{e14}) require $F[\tilde{H}(s)]$ and $G[\tilde{H}(s)]$ to maintain the structure of $F[H(s)]$ and $G[H(s)]$ given in (\ref{e10})-(\ref{e12}), which is not possible in general. Therefore, FWITIA may not satisfy the interpolation conditions (\ref{e13}) and (\ref{e14}) exactly, and thus $\bar{Y}\approx0$ and $\bar{Z}\approx0$ upon convergence.

 An interesting parallel should be noted that the nonzero cross-terms $P_{13}$, $P_{23}$, $Q_{14}$, and $Q_{24}$ inhibit FWBT to inherit the stability preservation and Hankel singular values retention properties of BT \citep{ghafoor2007frequency,sreeram2012improved}. Here also, the nonzero cross-terms inhibit the projection framework to satisfy the first-order optimality conditions exactly, which is possible in case of the standard $\mathcal{H}_2$-optimal MOR problem.
\subsection{Conditions for Exact Satisfaction of the Optimality Conditions}\label{sub3.2}
As seen already that to ensure $\bar{X}=0$, $\bar{Y}=0$, and $\bar{Z}=0$, we need to find the reduction matrices $\tilde{V}=P_{12}\tilde{P}^{-1}$ and $\tilde{W}=-Q_{12}\tilde{Q}^{-1}$ that ensures the oblique projection condition $\tilde{W}^T\tilde{V}=I$. By expanding the Lyapunov equations (\ref{e3}) and (\ref{e4}) according to the structure of $(A_w,B_w,C_w)$ in (\ref{e2}), it can be noted that $P_{23}$, $\tilde{P}$, $P_{12}$, $Q_{24}$, $\tilde{Q}$, and $Q_{12}$ solve the following Sylvester equations
\begin{align}
\tilde{A}P_{23}+P_{23}A_i^T+\tilde{B}(C_iP_i+D_iB_i^T)&=0,\label{e16}\\
\tilde{A}\tilde{P}+\tilde{P}\tilde{A}^T+\tilde{B}C_iP_{23}^T+P_{23}C_i^T\tilde{B}^T+\tilde{B}D_iD_i^T\tilde{B}^T&=0,\label{e17}\\
AP_{12}+P_{12}\tilde{A}^T+BC_iP_{23}^T+P_{13}C_i^T\tilde{B}^T+BD_iD_i^T\tilde{B}^T&=0,\label{e18}\\
\tilde{A}^TQ_{24}+Q_{24}A_o-\tilde{C}^T(B_o^TQ_o+D_o^TC_o)&=0,\label{e19}\\
\tilde{A}^T\tilde{Q}+\tilde{Q}\tilde{A}-\tilde{C}^TB_o^TQ_{24}^T-Q_{24}B_o\tilde{C}+\tilde{C}^TD_o^TD_o\tilde{C}&=0,\label{e20}\\
A^TQ_{12}+Q_{12}\tilde{A}+C^TB_o^TQ_{24}^T-Q_{14}B_o\tilde{C}-C^TD_o^TD_o\tilde{C}&=0.\label{e21}
\end{align}
If the ROM is obtained using the oblique projection $\Pi=-P_{12}\tilde{P}^{-1}\tilde{Q}^{-1}Q_{12}^T$ with $\tilde{V}=P_{12}\tilde{P}^{-1}$ and $\tilde{W}=-Q_{12}\tilde{Q}^{-1}$, (\ref{E2}) and the equations (\ref{e16})-(\ref{e21}) can be viewed as two coupled system of equations, i.e.,
\begin{align}
(\tilde{A},\tilde{B},\tilde{C})&=f(P_{12},Q_{12},\tilde{P},\tilde{Q}),\nonumber\\
(P_{12},Q_{12},\tilde{P},\tilde{Q})&=g(\tilde{A},\tilde{B},\tilde{C}).\nonumber
\end{align}
Clearly, the fixed points of $(\tilde{A},\tilde{B},\tilde{C})=f\big(g(\tilde{A},\tilde{B},\tilde{C})\big)$ ensure that $\bar{X}=0$, $\bar{Y}=0$, and $\bar{Z}=0$ if the condition for oblique projection $\tilde{W}^T\tilde{V}=I$ is satisfied. We now show in the next theorem that these fixed points satisfy the optimality conditions (\ref{e5})-(\ref{e7}) if $P_{13}-\tilde{V}P_{23}=0$ and $Q_{14}+\tilde{W}Q_{24}=0$ wherein $\tilde{V}$ and $\tilde{W}$ have full column rank.
\begin{theorem}\label{t1}
  Let $\tilde{A}$ be Hurwitz and $(\tilde{A},\tilde{B},\tilde{C})$ be a fixed point of $(\tilde{A},\tilde{B},\tilde{C})=f\big(g(\tilde{A},\tilde{B},\tilde{C})\big)$. Also, let that $\tilde{P}$ and $\tilde{Q}$ are invertible at the fixed point, and the fixed point is obtained by using the oblique projection $\Pi=-P_{12}\tilde{P}^{-1}\tilde{Q}^{-1}Q_{12}^T$ with $\tilde{V}=P_{12}\tilde{P}^{-1}$ and $\tilde{W}=-Q_{12}\tilde{Q}^{-1}$. Then $(\tilde{A},\tilde{B},\tilde{C})$ satisfies the first-order optimality conditions (\ref{e5})-(\ref{e7}) provided $P_{13}-\tilde{V}P_{23}=0$ and $Q_{14}+\tilde{W}Q_{24}=0$ wherein $\tilde{V}$ and $\tilde{W}$ have full column rank.
\end{theorem}
\begin{proof}
We need to show that when $P_{13}-\tilde{V}P_{23}=0$ and $Q_{14}+\tilde{W}Q_{24}=0$, $X=0$, $Y=0$, and $Z=0$ at the fixed points of $(\tilde{A},\tilde{B},\tilde{C})=f\big(g(\tilde{A},\tilde{B},\tilde{C})\big)$. By expanding the Lyapunov equation (\ref{e3}), one can note that $P_{34}$ and $P_{24}$ satisfy the following Sylvester equations
\begin{align}
A_iP_{34}+P_{34}A_o^T+(P_{13}^TC^T-P_{23}^T\tilde{C}^T)B_o^T&=0,\nonumber\\
\tilde{A}P_{24}+P_{24}A_o^T+\tilde{B}C_iP_{34}+(P_{12}^TC^T-\tilde{P}\tilde{C}^T)B_o^T&=0.\nonumber
\end{align} Since $\bar{Z}=0$ and $CP_{13}-\tilde{C}P_{23}=0$, we get
\begin{align}
A_iP_{34}+P_{34}A_o^T&=0&\textnormal{ and }&& \tilde{A}P_{24}+P_{24}A_o^T+\tilde{B}C_iP_{34}&=0.\nonumber
\end{align} Thus $P_{34}=0$ and $P_{24}=0$.

It can be noted by expanding the Lyapunov equation (\ref{e4}) that $Q_{34}$ and $Q_{23}$ solve the following Sylvester equations
\begin{align}
A_i^TQ_{34}+Q_{34}A_o+C_i^T(B^TQ_{14}+\tilde{B}^TQ_{24})&=0,\nonumber\\
\tilde{A}^TQ_{23}+Q_{23}A_i-\tilde{C}^TB_o^TQ_{34}^T+(\tilde{Q}\tilde{B}+Q_{12}^TB)C_i&=0.\nonumber
\end{align}
Now, since $\bar{Y}=0$ and $B^TQ_{14}+\tilde{B}^TQ_{24}=0$, we get
\begin{align}
A_i^TQ_{34}+Q_{34}A_o&=0&\textnormal{ and }&&\tilde{A}^TQ_{23}+Q_{23}A_i-\tilde{C}^TB_o^TQ_{34}^T&=0.\nonumber
\end{align} Thus $Q_{34}=0$, $Q_{23}=0$, and therefore, $X=0$.

Further, since $\tilde{W}^T\tilde{V}=I$, $\tilde{Q}=-Q_{12}^T\tilde{V}$ and  $\tilde{P}=\tilde{W}^TP_{12}$, the matrices $Y$ and $Z$ become
\begin{align}
Y&=Q_{12}^T(P_{13}-\tilde{V}P_{23})C_i^T=0,\nonumber\\
Z&= -B_o^T(Q_{14}^T+Q_{24}^T\tilde{W}^T)P_{12}=0.\nonumber
\end{align}
Thus $Y=0$ and $Z=0$. This completes the proof.
\end{proof}
\begin{remark}
When $W_i(s)=I$, $Y=0$, and the fixed point of $(\tilde{A},\tilde{B},\tilde{C})=f\big(g(\tilde{A},\tilde{B},\tilde{C})\big)$ satisfies the optimality condition (\ref{e6}) exactly if $\tilde{Q}$ is invertible at the fixed point. Similarly, when $W_o(s)=I$, $Z=0$, and the fixed point of $(\tilde{A},\tilde{B},\tilde{C})=f\big(g(\tilde{A},\tilde{B},\tilde{C})\big)$ satisfies the optimality condition (\ref{e7}) exactly if $\tilde{P}$ is invertible at the fixed point.
\end{remark}
Note that $\tilde{P}$ and $\tilde{Q}$ do not change the subspaces spanned by the columns of $\tilde{V}=P_{12}\tilde{P}^{-1}$ and $\tilde{W}=-Q_{12}\tilde{Q}^{-1}$ but only transform the basis of $P_{12}$ and $-Q_{12}$. Thus we can construct $\tilde{V}$ and $\tilde{W}$ as $\tilde{V}=P_{12}$ and $\tilde{W}=-Q_{12}$. By doing so, the invertibility of $\tilde{P}$ and $\tilde{Q}$ is no more required, which otherwise makes the problem quite restrictive because, during the process of finding the fixed points, $\tilde{P}$ and $\tilde{Q}$ may not be invertible all the time. This can cause numerical problems and limit the applicability of any possible fixed point iteration algorithm wherein the fixed points are obtained iteratively upon convergence. If the ROM is obtained by using the oblique projection $\Pi=-P_{12}Q_{12}^T$ with $\tilde{V}=P_{12}$ and $\tilde{W}=-Q_{12}$, (\ref{E2}) and the equations (\ref{e16})-(\ref{e21}) can be viewed as two coupled system of equations, i.e.,
\begin{align}
(\tilde{A},\tilde{B},\tilde{C})&=f_1(P_{12},Q_{12})&\textnormal{ and } &&(P_{12},Q_{12})&=g_1(\tilde{A},\tilde{B},\tilde{C}).\nonumber
\end{align}
In the next theorem, we show that the fixed points of $(\tilde{A},\tilde{B},\tilde{C})=f_1\big(g_1(\tilde{A},\tilde{B},\tilde{C})\big)$ satisfy the optimality conditions (\ref{e5})-(\ref{e7}) if $P_{13}-\tilde{V}P_{23}=0$ and $Q_{14}+\tilde{W}Q_{24}=0$ wherein $\tilde{V}$ and $\tilde{W}$ have full column rank.
\begin{theorem}
  Let $\tilde{A}$ be Hurwitz and $(\tilde{A},\tilde{B},\tilde{C})$ be a fixed point of $(\tilde{A},\tilde{B},\tilde{C})=f_1\big(g_1(\tilde{A},\tilde{B},\tilde{C})\big)$ obtained by using the oblique projection $\Pi=-P_{12}Q_{12}^T$ with $\tilde{V}=P_{12}$ and $\tilde{W}=-Q_{12}$. Then $(\tilde{A},\tilde{B},\tilde{C})$ satisfies the first-order optimality conditions (\ref{e5})-(\ref{e7}) provided $P_{13}-\tilde{V}P_{23}=0$ and $Q_{14}+\tilde{W}Q_{24}=0$ wherein $\tilde{V}$ and $\tilde{W}$ have full column rank.
\end{theorem}
\begin{proof}
Since the fixed points of $(\tilde{A},\tilde{B},\tilde{C})=f_1\big(g_1(\tilde{A},\tilde{B},\tilde{C})\big)$ are obtained by using the oblique projection $\Pi=-P_{12}Q_{12}^T$, the following holds $CP_{12}-\tilde{C}=0$, $Q_{12}^TB+\tilde{B}=0$, and $Q_{12}^TP_{12}+I=0$ at the fixed points. We first show that $\bar{X}=0$, $\bar{Y}=0$, and $\bar{Z}=0$ at the fixed points of $(\tilde{A},\tilde{B},\tilde{C})=f_1\big(g_1(\tilde{A},\tilde{B},\tilde{C})\big)$ if $P_{13}-\tilde{V}P_{23}=0$ and $Q_{14}+\tilde{W}Q_{24}=0$ wherein $\tilde{V}$ and $\tilde{W}$ have full column rank. By multiplying (\ref{e18}) with $\tilde{W}^T$ from the left, we get
\begin{align}
\tilde{W}^TAP_{12}+\tilde{W}^TP_{12}\tilde{A}^T+\tilde{W}^TBC_iP_{23}^T+\tilde{W}^TP_{13}C_i^T\tilde{B}^T+\tilde{W}^TBD_iD_i^T\tilde{B}^T&=0.\nonumber
\end{align}
Note that $\tilde{W}^TP_{12}=\tilde{W}^T\tilde{V}=I$. Also, note that $\tilde{W}^TP_{13}=P_{23}$, since $P_{13}-\tilde{V}P_{23}=0$. Thus
\begin{align}
\tilde{A}+\tilde{A}^T+\tilde{B}C_iP_{23}^T+P_{23}C_i^T\tilde{B}^T+\tilde{B}D_iD_i^T\tilde{B}^T&=0.\nonumber
\end{align}
Due to uniqueness, $\tilde{P}=I$, and thus $\bar{Z}=0$.

By multiplying (\ref{e21}) with $\tilde{V}^T$ from the left, we get
\begin{align}
\tilde{V}^TA^TQ_{12}+\tilde{V}^TQ_{12}\tilde{A}+\tilde{V}^TC^TB_o^TQ_{24}^T-\tilde{V}^TQ_{14}B_o\tilde{C}-\tilde{V}^TC^TD_o^TD_o\tilde{C}&=0.\nonumber
\end{align}
Note that $\tilde{V}^TQ_{14}=-Q_{24}$ since $Q_{14}+\tilde{W}Q_{24}=0$. Also, note that $\tilde{V}^T\tilde{W}=-\tilde{V}^TQ_{12}=I$. Thus
\begin{align}
-\tilde{A}^T-\tilde{A}+\tilde{C}^TB_o^TQ_{24}^T+Q_{24}B_o\tilde{C}-\tilde{C}^TD_o^TD_o\tilde{C}&=0.\nonumber
\end{align}
Due to uniqueness, $\tilde{Q}=I$, and thus $\bar{Y}=0$ and $\bar{X}=0$.

It is now left to show that $X=0$, $Y=0$, and $Z=0$. From Theorem \ref{t1}, we know that when $\bar{Y}=0$, $\bar{Z}=0$, $P_{13}=\tilde{V}P_{23}$, and $Q_{14}=-\tilde{W}Q_{24}$, the matrices $P_{34}=0$, $Q_{23}=0$ and $X=0$. Further, since $\tilde{P}=I$, and $\tilde{Q}=I$, $Y$ and $Z$ become
\begin{align}
Y&=(Q_{12}^TP_{13}+P_{23})C_i^T&\textnormal{ and }&&Z&=B_o^T(Q_{14}^TP_{12}+Q_{24}^T).\nonumber
\end{align} Since $P_{13}-\tilde{V}P_{23}=0$ and $Q_{14}+\tilde{W}Q_{24}=0$, the matrices $Y=0$ and $Z=0$. This completes the proof.
\end{proof}
\begin{remark}
For the oblique projection $\Pi=-P_{12}\tilde{P}^{-1}\tilde{Q}^{-1}Q_{12}^T$, the fixed points of $(\tilde{A},\tilde{B},\tilde{C})=f\big(g(\tilde{A},\tilde{B},\tilde{C})\big)$ ensures $\bar{X}=0$, $\bar{Y}=0$, and $\bar{Z}=0$ regardless of whether the conditions $P_{13}=\tilde{V}P_{23}$ and $Q_{14}=-\tilde{W}Q_{24}$ hold or not. However, for the oblique projection $\Pi=-P_{12}Q_{12}^T$, the fixed points of $(\tilde{A},\tilde{B},\tilde{C})=f_1\big(g_1(\tilde{A},\tilde{B},\tilde{C})\big)$ ensures $\bar{X}=0$, $\bar{Y}=0$, and $\bar{Z}=0$ only if $P_{13}=\tilde{V}P_{23}$ and $Q_{14}=-\tilde{W}Q_{24}$ also hold. Therefore, although the reduction matrices $\tilde{V}$ and $\tilde{W}$ span the same subspace in both cases, the change of basis in the latter case incurs deviations in $\bar{X}=0$, $\bar{Y}=0$, and $\bar{Z}=0$ if the conditions $P_{13}=\tilde{V}P_{23}$ and $Q_{14}=-\tilde{W}Q_{24}$ are violated.
\end{remark}

By expanding the Lyapunov equations (\ref{e3}) and (\ref{e4}), one can note that $P_{13}$ and $Q_{14}$ solve the following Sylvester equations
\begin{align}
AP_{13}+P_{13}A_i^T+B(C_iP_i+D_iB_i^T)&=0,\label{e22}\\
A^TQ_{14}+Q_{14}A_o+C^T(B_o^TQ_o+D_o^TC_o)&=0.\label{e23}
\end{align}
Thus $\tilde{V}P_{23}$ and $-\tilde{W}Q_{24}$ can be seen as Petrov-Galerkin approximations of $P_{13}$ and $Q_{14}$ in the following sense
\begin{align}
Ran\big(A\tilde{V}P_{23}+\tilde{V}P_{23}A_i^T+B(C_iP_i+D_iB_i^T)\big)&\bot Ran\big(\tilde{W}\big),\nonumber\\
Ran\big(A^T\tilde{W}Q_{24}+\tilde{W}Q_{24}A_o-C^T(B_o^TQ_o+D_o^TC_o)\big)&\bot Ran\big(\tilde{V}\big).\nonumber
\end{align}

In general, $P_{13}\neq\tilde{V}P_{23}$ and $Q_{14}\neq-\tilde{W}Q_{24}$. To achieve a nearly optimum ROM, we need to find fixed points of $(\tilde{A},\tilde{B},\tilde{C})=f_1\big(g_1(\tilde{A},\tilde{B},\tilde{C})\big)$ by using the oblique projection $\Pi=\tilde{V}\tilde{W}^T$, which also provides good Petrov-Galerkin approximations of $P_{13}$ and $Q_{14}$ as $\tilde{V}P_{23}\approx P_{13}$ and $-\tilde{W}Q_{24}\approx Q_{14}$.
\begin{remark}
When $W_i(s)$ and $W_o(s)$ are co-inner and inner functions, respectively, the matrices $P_{13}=0$, $P_{23}=0$, $Q_{14}=0$, and $Q_{24}=0$ \citep{sahlan2007properties,sreeram2002properties,sreeram2012improved}. Thus $P_{13}-\tilde{V}P_{23}=0$ and $Q_{14}+\tilde{W}Q_{24}=0$, and the fixed points of $(\tilde{A},\tilde{B},\tilde{C})=f_1\big(g_1(\tilde{A},\tilde{B},\tilde{C})\big)$ satisfy the optimality conditions (\ref{e5})-(\ref{e7}) exactly.
\end{remark}
\subsection{Deviation in the Optimality Conditions}
We now show that as the order of the ROM increases, the fixed points of $(\tilde{A},\tilde{B},\tilde{C})=f_1\big(g_1(\tilde{A},\tilde{B},\tilde{C})\big)$ implicitly ensures that $P_{13}\approx\tilde{V}P_{23}$ and $Q_{14}\approx-\tilde{W}Q_{24}$. Thus the deviation in the satisfaction of the optimality conditions (\ref{e5})-(\ref{e7}) decays as the order of ROM increases.

To observe this, note that
\begin{align}
||E(s)W_i(s)||_{\mathcal{H}_2}^2&=trace(CPC^T-2CP_{12}\tilde{C}^T+\tilde{C}\tilde{P}\tilde{C}^T).\nonumber
\end{align} When $\bar{Z}=0$, $||E(s)W_i(s)||_{\mathcal{H}_2}^2$ becomes
\begin{align}
||E(s)W_i(s)||_{\mathcal{H}_2}^2&=trace(CPC^T-\tilde{C}\tilde{P}\tilde{C}^T)=trace\big(C\big(P-\tilde{V}\tilde{P}\tilde{V}\big)C^T\big).\nonumber
\end{align} Thus as the order of ROM increases and $||E(s)W_i(s)||_{\mathcal{H}_2}^2$ decreases, $\hat{P}=\tilde{V}\tilde{P}\tilde{V}$ approaches $P$. Also, since
\begin{align}
||W(s)E(s)||_{\mathcal{H}_2}^2&=trace(B^TQB+2B^TQ_{12}\tilde{B}+\tilde{B}^T\tilde{Q}\tilde{B}),\nonumber
\end{align} $\bar{Y}=0$, $||W_o(s)E(s)||_{\mathcal{H}_2}^2$ becomes
\begin{align}
||W_o(s)E(s)||_{\mathcal{H}_2}^2&=trace(B^TQB-\tilde{B}^T\tilde{Q}\tilde{B})=trace\big(B^T\big(Q-\tilde{W}\tilde{Q}\tilde{W}^T\big)B\big).\nonumber
\end{align} As the order of ROM increases and $||W_o(s)E(s)||_{\mathcal{H}_2}^2$ decreases, $\hat{Q}=\tilde{W}\tilde{Q}\tilde{W}^T$ approaches $Q$.

$P$ and $Q$ solve the following Lyapunov equations
\begin{align}
AP+PA^T+BC_iP_{13}^T+P_{13}C_i^TB^T+BD_iD_i^TB^T&=0,\label{e24}\\
A^TQ+QA+C^TB_o^TQ_{14}^T+Q_{14}B_oC+C^TD_o^TD_oC&=0.\label{e25}
\end{align}
Let the residuals $R_1$ and $R_2$ be defined as
\begin{align}
R_1&=A\hat{P}+\hat{P}A^T+BC_iP_{13}^T+P_{13}C_i^TB^T+BD_iD_i^TB^T,\nonumber\\
R_2&=A^T\hat{Q}+\hat{Q}A+C^TB_o^TQ_{14}^T+Q_{14}B_oC+C^TD_o^TD_oC.\nonumber
\end{align}
As $\hat{P}$ and $\hat{Q}$ approach $P$ and $Q$, respectively, $R_1$ and $R_2$ approach zero. Further, when $R_1\approx0$ and $R_2\approx0$, the Petrov-Galerkin conditions $\tilde{W}^TR_1\tilde{W}\approx0$ and $\tilde{V}^TR_2\tilde{V}\approx0$ also hold approximately, which imply that $\tilde{W}^TP_{13}\approx P_{23}$ and $\tilde{V}^TQ_{14}\approx -Q_{24}$. The singular values of $P$ and $Q$ decay rapidly in the weighted case \citep{benner2016frequency,kurschner2018balanced}. Thus $||R_1||$ and $||R_2||$ are expected to decay quickly for a relatively smaller value of $r$ due to low numerical rank of $P$ and $Q$. Therefore, the conditions $\tilde{V}P_{23}\approx P_{13}$ and $-\tilde{W}Q_{24}\approx Q_{14}$ are expected to be met without having to increase the value of $r$ too much. In short, a compact ROM that nearly satisfies the optimality conditions (\ref{e5})-(\ref{e7}) can be obtained with the oblique projection $\Pi=-P_{12}Q_{12}^T$.
\section{Frequency-weighted $\mathcal{H}_2$-suboptimal MOR}\label{sec3}
In this section, a fixed point iteration algorithm is proposed, which on convergence tends to satisfy $\bar{X}=0$, $\bar{Y}=0$, and $\bar{Z}=0$, and therefore the resulting ROM tends to satisfy the optimality conditions (\ref{e5})-(\ref{e7}).
\subsection{Fixed-point Iteration Algorithm}
The fixed points of $(\tilde{A},\tilde{B},\tilde{C})=f_1\big(g_1(\tilde{A},\tilde{B},\tilde{C})\big)$ can be found by using the fixed point iteration algorithm with an additional constraint that $P_{12}$ and $-Q_{12}$ satisfy the oblique projection condition $-Q_{12}^TP_{12}=I$. To ensure that $\tilde{W}^T\tilde{V}=I$, most of the $\mathcal{H}_2$-optimal MOR algorithms use the correction equation $\tilde{W}=\tilde{W}(\tilde{V}^T\tilde{W})^{-1}$. Theoretically, it does ensure that $\tilde{W}^T\tilde{V}=I$, however, it becomes numerically unstable even for small systems \citep{bennersparse}. A more robust approach is to take the geometric interpretation of $\tilde{W}^T\tilde{V}=I$, i.e., the columns of $\tilde{W}$ and $\tilde{V}$ form biorthogonal basis of a subspace in $\mathbb{R}^n$ \citep{bennersparse}. Therefore, we use biorthogonal Gram-Schmidt method (steps \ref{st6}-\ref{st11} of Algorithm \ref{Alg1}) to ensure that $\tilde{W}^T\tilde{V}=I$ for better numerical properties. The pseudo code of our approach is given in Algorithm \ref{Alg1}, which is referred to as the frequency-weighted $\mathcal{H}_2$-suboptimal MOR algorithm (FWHMOR).
\begin{algorithm}[!b]
\textbf{Input:} Original system: $(A,B,C,D)$; Input weight: $(A_i,B_i,C_i,D_i)$; Output weight: $(A_o,B_o,C_o,D_o)$, Initial guess: $(\tilde{A},\tilde{B},\tilde{C})$.\\
\textbf{Output:} ROM $(\tilde{A},\tilde{B},\tilde{C})$.
  \begin{algorithmic}[1]
      \STATE Compute $P_i$ and $Q_o$ by solving\label{st1}
      \begin{align}
      A_iP_i+P_iA_i^T+B_iB_i^T&=0,\nonumber\\
      A_o^TQ_o+Q_oA_o+C_o^TC_o&=0.\nonumber
      \end{align}
       \STATE Compute $P_{13}$ and $Q_{14}$ from the equations (\ref{e22}) and (\ref{e23}), respectively.\label{st2}
      \STATE \textbf{while} (not converged) \textbf{do}
      \STATE Compute $P_{23}$ and $Q_{24}$ from the equations (\ref{e16}) and (\ref{e19}), respectively.\label{st4}
      \STATE Compute $P_{12}$ and $Q_{12}$ from the equations (\ref{e18}) and (\ref{e21}), respectively.\label{st5}
      \STATE \textbf{for} $i=1,\ldots,r$ \textbf{do}\label{st6}
      \STATE $v=P_{12}(:,i)$, $v=\prod_{k=1}^{i}\big(I+P_{12}(:,k)Q_{12}(:,k)^T\big)v$.
      \STATE $w=-Q_{12}(:,i)$, $w=\prod_{k=1}^{i}\big(I+Q_{12}(:,k)P_{12}(:,k)^T\big)w$.
      \STATE $v=\frac{v}{||v||_2}$, $w=\frac{w}{||w||_2}$, $v=\frac{v}{w^Tv}$.
      \STATE $\tilde{V}(:,i)=v$, $\tilde{W}(:,i)=w$.
      \STATE \textbf{end for}\label{st11}
      \STATE $\tilde{A}=\tilde{W}^TA\tilde{V}$, $\tilde{B}=\tilde{W}^TB$, $\tilde{C}=C\tilde{V}$.
      \STATE \textbf{end while}
  \end{algorithmic}
  \caption{FWHMOR}\label{Alg1}
\end{algorithm}
\begin{remark}
FWHMOR provides approximations of $P$ and $Q$ as $\hat{P}$ and $\hat{Q}$, respectively, which can be used in FWBT to save some computational cost by avoiding the computation of large-scale Lyapunov equations (\ref{e24}) and (\ref{e25}).
\end{remark}
\subsection{Connection with FWITIA\label{3.3.1}}
In this subsection, we show that FWHMOR satisfies the tangential interpolation conditions (\ref{e13}) and (\ref{e14}) under some conditions. To begin with, the connection between the reduction matrices of FWHMOR and FWITIA is established in Proposition \ref{prop4.1}.
\begin{proposition}\label{prop4.1}
Suppose FWITIA and FWHMOR have converged, and $(V_a,W_a)$ and $(\tilde{V},\tilde{W})$ are the respective reduction matrices of the last iteration. Also, suppose $H(s)$ and $\tilde{H}(s)$ have simple poles. Then the columns of $V_a$ and $W_a$ span the same subspaces as that spanned by the columns of $\tilde{V}$ and $\tilde{W}$, respectively.
\end{proposition}
\begin{proof}
Let us denote the spectral factorization of $\tilde{A}$ as $\tilde{A}=RSR^{-1}$ where $S=diag(\tilde{\lambda}_1,\cdots,\tilde{\lambda_r})$. Now define $L_i$ and $L_o$ as $L_i=\tilde{B}^TR^{-*}=\begin{bmatrix}\tilde{r}_1&\cdots&\tilde{r}_r\end{bmatrix}$ and $L_o=\tilde{C}R=\begin{bmatrix}\tilde{l}_1&\cdots&\tilde{l}_r\end{bmatrix}$, respectively. Owing to the connection of Sylvester equations and rational Krylov subspaces \citep{panzer2014model,wolf2014h}, it is shown in \citep{zulfiqar2019frequency} that $V_a$ and $W_a$ in FWITIA satisfy the following Sylvester equations
\begin{align}
AV_a+V_aS^*+BC_iV_b^T+P_{13}C_i^TL_i+BD_iD_i^TL_i&=0,\nonumber\\
A^TW_a+W_aS+C^TB_i^TW_b^T+Q_{14}B_iL_o+C^TD_o^TD_oL_o&=0\nonumber
\end{align}
where
\begin{align}
SV_b+V_bA_i^T+L_i^*(C_iP_i+D_iB_i^T)&=0,\nonumber\\
S^*W_b+W_bA_o+L_o^T(B_o^TQ_o+D_o^TC_o)&=0.\nonumber
\end{align}
By putting $\tilde{A}=RSR^{-1}$ in (\ref{e16}) and (\ref{e18}), pre-multiplying (\ref{e16}) with $R^{-1}$, and post-multiplying (\ref{e18}) with $R^{-*}$, one can note that the following Sylvester equations hold
\begin{align}
SR^{-1}P_{23}+R^{-1}P_{23}A_i^T+L_i^*(C_iP_i+D_iB_i^T)&=0,\nonumber\\
AP_{12}R^{-*}+P_{12}R^{-*}S^*+BC_iP_{23}^TR^{-*}+P_{13}C_i^TL_i+BD_iD_i^TL_i&=0.\nonumber
\end{align}
Due to uniqueness, $V_a=P_{12}R^{-*}$ and $V_b=R^{-1}P_{23}$. Similarly, by putting $\tilde{A}=RSR^{-1}$ (\ref{e19}) and (\ref{e21}), pre-multiplying (\ref{e19}) with $R^{*}$, and post-multiplying (\ref{e21}) with $R^{-*}$, one can note that the following Sylvester equations hold
\begin{align}
S^*R^{*}Q_{24}+R^{*}Q_{24}A_o-L_o^T(B_o^TQ_o+D_o^TC_o)&=0,\nonumber\\
A^TQ_{12}R+Q_{12}RS-C^TB_i^TQ_{24}^TR-Q_{14}B_iL_o-C^TD_o^TD_oL_o&=0.\nonumber
\end{align}
Due to uniqueness, $W_a=-Q_{12}R$ and $W_b=-R^{*}Q_{24}$. Since $R$ only changes the basis of $V_a$ and $W_a$, the columns of $\tilde{V}$ and $\tilde{W}$ in FWHMOR span the same subspaces as spanned by $V_a$ and $W_a$, respectively, in FWITIA.
\end{proof}
From Proposition \ref{prop4.1}, it is clear that the ROM constructed by FWHMOR satisfies the tangential interpolation conditions (\ref{e13}) and (\ref{e14}) upon convergence like FWITIA, provided $P_{13}=\tilde{V}P_{23}$ and $Q_{14}=-\tilde{W}Q_{24}$. However, there are some notable numerical differences between FWHMOR and FWITIA. FWHMOR does not require $H(s)$ and $\tilde{H}(s)$ to have simple poles, unlike FWITIA, to construct the local optimum. Thus $\tilde{A}$ does not need to be diagonalizable in FWHMOR. Therefore, FWITIA can be considered equivalent to FWHMOR if $H(s)$ and $\tilde{H}(s)$ have simple poles. The spectral factorization of $\tilde{A}$ in every iteration of FWITIA may cause numerical ill-conditioning \citep{bennersparse}. Moreover, FWITIA uses the correction equation $\tilde{W}=\tilde{W}(\tilde{V}^T\tilde{W})^{-1}$ to ensure the oblique projection condition $\tilde{W}^T\tilde{V}=I$, whereas FWHMOR uses numerically more stable biorthogonal Gram-Schmidt \citep{bennersparse} to achieve that. In short, FWHMOR is numerically more general and stable algorithm than FWITIA, though both span the same subspaces. Moreover, the results of Section \ref{sec2} provide the theoretical foundation for FWITIA in terms of seeking to satisfy the optimality conditions (\ref{e5})-(\ref{e7}). Hence, FWITIA is no more a heuristic generalization of \citep{van2008h2} but an interpolation framework for the frequency-weighted $\mathcal{H}_2$-optimal MOR problem.
\subsection{Computational Aspects}
We now discuss some computational aspects of FWHMOR to be considered for its efficient numerical implementation.
\subsubsection{Initial Guess} The initial guess of the ROM can be made arbitrarily, for instance, by direct truncation of the original state-space realization. However, a good choice of the initial ROM generally has a positive impact on the performance of the fixed point iteration methods. Therefore, it is recommended to generate the initial guess by using the eigensolver proposed in \citep{rommes2006efficient}. Since the mirror images of the poles with large residues have a big contribution to the $\mathcal{H}_2$-norm, the initial guess can be generated with the eigensolver proposed in \citep{rommes2006efficient} by projecting $H(s)$ onto the dominant eigenspace of $A$. Another option is to compute the initial ROM by using the low-rank approximation methods in \citep{ahmad2010krylov,benner2014computing} such that it provides good Petrov-Galerkin approximations of $P_{13}$ and $Q_{14}$ as $\tilde{V}P_{23}$ and $-\tilde{W}Q_{24}$. This ensures that $||P_{13}-\tilde{V}P_{23}||$ and $||Q_{14}+\tilde{W}Q_{24}||$ are small to begin with.
\subsubsection{Convergence and Stopping Criteria}
Like in most of the $\mathcal{H}_2$-optimal MOR algorithms, the convergence is not guaranteed in FWHMOR. Therefore, a good stopping criterion is required to stop the algorithm in case it does not converge within admissible time. The stopping criterion should have two main properties: (i) It should be easily computable (ii) It should quickly indicate that the error has dropped appreciably. These two properties make sure that the computation of stopping criteria is not a computational burden in itself, and it can save computational effort by indicating that the algorithm is not improving the accuracy of ROM any further. Owing to connection between FWITIA and FWHMOR, the relative change in eigenvalues of $\tilde{A}$ can be used as the stopping criterion. Due to the small size of $\tilde{A}$, this can be achieved accurately and cheaply using $QZ$-method. $||\bar{X}||_2$ can also be used as a stopping criterion. The computation of $\tilde{P}$ and $\tilde{Q}$ in $||\bar{X}||_2$ requires solutions of two small-scale Lyapunov equations, i.e., (\ref{e17}) and (\ref{e20}), which can be done cheaply.  Also, note that from a pragmatic perspective, achieving less $||E_w(s)||_{\mathcal{H}_2}^2$ is the main objective and not the local optimum in itself. Thus the stopping criterion can be based directly on the error itself. However, the computation of $||E_w(s)||_{\mathcal{H}_2}^2$ in each iteration is an expensive operation in a large-scale setting. We have discussed in Section \ref{sec2} that $\bar{Y}=0$ and $\bar{Z}=0$ essentially minimize $||W(s)E(s)||_{\mathcal{H}_2}^2$ and $||E(s)V(s)||_{\mathcal{H}_2}^2$, respectively. Therefore, one can use the relative changes in $e_1=tr(2CP_{12}\tilde{C}^T-\tilde{C}\tilde{P}\tilde{C}^T)$ and $e_2=tr(-2B^TQ_{12}\tilde{B}-\tilde{B}^T\tilde{Q}\tilde{B})$ as the stopping criteria. The algorithm can be stopped when relative changes in $e_1$ and $e_2$ stagnate because this indicates that $||E(s)V(s)||_{\mathcal{H}_2}^2$ and $||W(s)E(s)||_{\mathcal{H}_2}^2$ are not changing. Another criterion for stopping the algorithm can be the number of iterations or computational time. If the other stopping criteria are not achieved within the maximum allowable number of iterations or admissible time, the algorithm can be stopped prematurely.
\subsubsection{Computational Cost}
The computational cost of FWHMOR depends on several factors. In step \ref{st1}, $P_i$ and $Q_o$ can be computed cheaply if $n_i$ and $n_o$ are small. However, if the weights $W_i(s)$ and $W_o(s)$ are large-scale transfer functions, $P_i$ and $Q_o$ should be replaced with their low-rank approximations, for instance, by using the toolboxes \citep{penzl1999lyapack,saak2010matrix}. In step \ref{st2}, $P_{13}$ and $Q_{14}$ can be computed within admissible time if $n_i$ and $n_o$ are small due to the \textit{sparse-dense} structure of the Sylvester equations (\ref{e22}) and (\ref{e23}) \citep{panzer2014model,wolf2014h}. The computational effort can further be reduced by using the efficient algorithm proposed in \citep{bennersparse} for this kind of Sylvester equations. The linear system of equations in \citep{bennersparse} can be solved by using sparse solvers like \citep{castagnotto2017sss,davis2004algorithm,demmel1999asynchronous,demmel1999supernodal} to further save the computational cost. If the weights $W_i(s)$ and $W_o(s)$ are large-scale transfer functions, then $P_{13}$ and $Q_{14}$ should also be replaced with their low-rank approximations like $P_i$ and $Q_o$. In steps \ref{st4} and \ref{st5}, $P_{23}$, $Q_{24}$, $P_{12}$, and $Q_{12}$ solve \textit{sparse-dense} Sylvester equations, which can again be solved by using the solver in \citep{bennersparse}.
\section{Numerical Results}\label{sec4}
In this section, FWHMOR is tested on four numerical examples. The first example is an illustrative one, which is presented to aid convenient repeatability and validation of all the theoretical results of the paper. The second and third examples are frequency-weighted MOR problems, and the fourth example is a controller reduction problem. The original high-order models in the last three examples are taken from the benchmark collection for MOR of \citep{chahlaoui2005benchmark}. Although the ROM constructed by FWBT is not optimal in any norm, it offers supreme accuracy and is considered a gold standard for the frequency-weighted MOR problem \citep{ghafoor2008survey}. Therefore, we compare the performance of our algorithm with FWBT. Further, we replace $P$ and $Q$ in FWBT with $\hat{P}$ and $\hat{Q}$, respectively, to perform approximate  FWBT, which we refer to as Approximate-FWBT (A-FWBT) wherein $\hat{P}$ and $\hat{Q}$ are generated by FWHMOR.
\subsection{Experimental Setup and Hardware:} In all examples, FWHMOR is initialized arbitrarily, and the mirror images of the poles of the initial guess used in FWHMOR are selected as interpolation points. For the multi-input multi-output (MIMO) example, the residues of the initial guess used in FWHMOR are selected as tangential directions in FWITIA. This ensures a fair comparison between FWITIA and FWHMOR, as both algorithms are expected to behave similarly with this selection. The relative change in the poles of the ROM is used as the stopping criterion with a tolerance of $1\times 10^{-2}$. The Lyapunov and Sylvester equations are solved using MATLAB's \textit{`lyap'} command. The experiments are performed using MATLAB $2016$ on a computer with a $2$GHz $i7$ processor, $16$GB random access memory, and Windows $10$ operating system.
\subsection{Illustrative Example} Consider a $6^{th}$ order system with the following state-space realization
\begin{align}
A&=\begin{bsmallmatrix}
         0    &     0    &     0   & 1    &     0    &     0\\
         0     &    0     &    0    &     0  &  1 &         0\\
         0     &    0     &    0    &     0   &      0  &  1\\
   -5.4545  &  4.5455    &     0 &  -0.0545  &  0.0455   &      0\\
   10 & -21 &   11 &   0.1 &  -0.21 &   0.11 \\
         0  &  5.5 &   -6.5 &        0  &  0.055 &   -0.065\end{bsmallmatrix},&B&=\begin{bsmallmatrix}0 & 0 &0 &0.0909 & 0.4 &   -0.5\end{bsmallmatrix}^T,\nonumber\\
         C&=\begin{bsmallmatrix}2  & -2  & 3  & 0  & 0  & 0\end{bsmallmatrix}.\nonumber
\end{align}
Let the input and output frequency weights be the following
\begin{align}
A_i&=\begin{bmatrix}-2  & -4.375\\ 8  &    0\end{bmatrix}, &B_i&=\begin{bmatrix}2  &   0\end{bmatrix}^T,& C_i&=\begin{bmatrix} 1&0\end{bmatrix},\nonumber\\
A_o&=\begin{bmatrix} -5 & -9.375\\ 16 & 0\end{bmatrix}, &B_o&=\begin{bmatrix}2  &   0\end{bmatrix}^T,& C_o&=\begin{bmatrix} 2.5&0\end{bmatrix}.\nonumber
\end{align}
The initial guess used in FWHMOR is the following
\begin{align}
\tilde{A}^{(0)}&=\begin{bmatrix}0.0332 &   5.4109\\-4.8283   -0.2998\end{bmatrix},&\tilde{B}^{(0)}&=\begin{bmatrix}-0.0747&-0.2958\end{bmatrix}^T,\nonumber\\
\tilde{C}^{(0)}&=\begin{bmatrix}1.0117& -0.2599\end{bmatrix}.\nonumber
\end{align}
Both FWHMOR and FWITIA converge in $4$ iterations. The reduction matrices in FWHMOR are the following
\begin{align}
\tilde{V}&=\begin{bmatrix}0.2132&-0.0046\\-0.9666&0.0623\\0.2671&-0.0341\\0.14&0.348\\-1.3601&-1.6698\\0.6732&0.5017\end{bmatrix}&\textnormal{ and } &&\tilde{W}&=\begin{bmatrix}0.4167&-0.337\\-0.7398&0.6556\\0.5269&-0.3808\\0.0139&0.2199\\-0.0325&-0.4411\\0.0137&0.2622\end{bmatrix},\nonumber
\end{align}
which construct the following ROM
\begin{align}
\tilde{A}&=\begin{bmatrix}0.4059& 1.6956\\-15.6668&-0.6719\end{bmatrix}, &\tilde{B}&=\begin{bmatrix}-0.0186&-0.2875\end{bmatrix}^T,\nonumber\\
\tilde{C}&=\begin{bmatrix}3.1608&-0.2362\end{bmatrix}.\nonumber
\end{align}
The reduction matrices in FWITIA are the following
\begin{align}
\tilde{V}&=\begin{bmatrix}0.0086&0.1932\\-0.063&-0.9053\\0.0272&0.2621\\-0.1994&-0.1223\\0.9381&-0.0251\\-0.2746&0.2427\end{bmatrix}&\textnormal{ and } &&\tilde{W}&=\begin{bmatrix}0.0084&0.4674\\-0.1161&-0.827\\-0.0698&0.5935\\-0.4109& 0.0283\\0.8306& -0.0621\\-0.4857& 0.0304\end{bmatrix},\nonumber
\end{align}
which construct the following ROM
\begin{align}
\tilde{A}&=\begin{bmatrix}1.4570&25.1669\\-1.1444&-1.7230\end{bmatrix}, &\tilde{B}&=\begin{bmatrix} 0.5377&-0.0374\end{bmatrix}^T,\nonumber\\
\tilde{C}&=\begin{bmatrix} 0.2248&2.9833\end{bmatrix}.\nonumber
\end{align}
One can verify by using MATLAB's command \textit{T = mldivide(V,V1)} that the reduction matrices and the ROMs generated by FWITIA and FWHMOR are related to each other with the similarity transformation $T=\begin{bmatrix}0.0275&0.8906\\-0.5842&-0.7104\end{bmatrix}$. This numerically confirms the results of Subsection \ref{3.3.1}. The deviations in the optimality conditions (\ref{e5})-(\ref{e7}) and the interpolation conditions (\ref{e13}) and (\ref{e14}) (which are denoted by $\mathscr{F}$ and $\mathscr{G}$, respectively) for both ROMs are tabulated in Table \ref{tab1}. It can be noted that the deviations are so small that these ROMs can be considered as local optima for all practical purposes. Moreover, FWITIA and FWHMOR also provide good approximations of $P$ and $Q$.
\begin{table}[!h]
\centering
\caption{Deviation in the optimality conditions}\label{tab1}
\begin{tabular}{|c|c|c|}
\hline
Deviation & FWITIA&FWHMOR\\ \hline
 $||\bar{X}+X||_2$         &$2.90\times 10^{-4}$ &     $1.88\times 10^{-4}$   \\ \hline
  $||\bar{Y}D_iD_i^T+Y||_2$        &   $1.19\times 10^{-4}$   &  $1.06\times 10^{-4}$\\ \hline
  $||D_o^TD_o\bar{Z}+Z||_2$        &   $2.26\times 10^{-5}$   & $1.46\times 10^{-5}$ \\ \hline
  $||\mathscr{F}||_2$        &   $6.96\times 10^{-4}$   & $6.96\times 10^{-4}$ \\ \hline
   $||\mathscr{G}||_2$        &   $2.13\times 10^{-5}$   & $2.13\times 10^{-5}$ \\ \hline
  $||P_{13}-\tilde{V}P_{23}||_2$        &  $0.0946$    & $0.0946$ \\ \hline
  $||Q_{14}+\tilde{W}Q_{24}||_2$        &   $0.1097$   &  $0.1096$\\ \hline
   $||P-\tilde{V}\tilde{P}\tilde{V}^T||_2$        &   $0.0419$   &  $0.0419$\\ \hline
   $||Q-\tilde{W}\tilde{Q}\tilde{W}^T||_2$        &   $0.2247$   &  $0.2247$\\ \hline
\end{tabular}
\end{table}
The $\mathcal{H}_2$- and $\mathcal{H}_\infty$-norms of the weighted error transfer function $E_w(s)$ are compared with FWBT in Table \ref{tab2}. It can be noted that FWHMOR and A-FWBT provide good approximation.
\begin{table}[!h]
\centering
\caption{Weighted Error}\label{tab2}
\begin{tabular}{|c|c|c|}
\hline
Technique &$||E_w(s)||_{\mathcal{H}_2}$&$||E_w(s)||_{\mathcal{H}_\infty}$\\ \hline
 FWBT         &$0.0080$ &     $0.0471$   \\ \hline
 FWITIA        &   $0.0061$   &  $0.0471$\\ \hline
 FWHMOR        &   $0.0061$   & $0.0471$ \\ \hline
 A-FWBT       &  $0.0061$    & $0.0471$ \\ \hline
\end{tabular}
\end{table}
\subsection{Clamped Beam} Consider the $348^{th}$ order clamped beam model from the benchmark collection of \citep{chahlaoui2005benchmark}. Suppose a ROM of the clamped beam model is required, which ensures high fidelity within the frequency interval $[5,25]$ rad/sec. To achieve good accuracy within the desired frequency interval, a $4^{th}$ order band-pass filter with the passband $[5,10]$ rad/sec is used as the input weight, which is designed by using MATLAB's command \textit{butter(2,[5,10],'s')}. Moreover, a $4^{th}$ order band-pass filter with the passband $[10,25]$ rad/sec is used as the output weight, which is designed by using MATLAB's command \textit{butter(2,[10,25],'s')}. A $5^{th}$ order ROM is obtained by using FWBT, FWITIA, FWHMOR, and A-FWBT. The $\mathcal{H}_2$- and $\mathcal{H}_\infty$-norms of the weighted error transfer function $E_w(s)$ are tabulated in Table \ref{tab3}. It can be seen that FWHMOR and A-FWBT construct accurate ROMs. The singular values of $E(s)$ within $[5,25]$ rad/sec are plotted in Figure \ref{fig1}. It can be seen that FWHMOR and A-FWBT ensure good accuracy within the desired frequency region.
\begin{table}[!h]
\centering
\caption{Weighted Error}\label{tab3}
\begin{tabular}{|c|c|c|}
\hline
Technique &$||E_w(s)||_{\mathcal{H}_2}$&$||E_w(s)||_{\mathcal{H}_\infty}$\\ \hline
 FWBT         &$0.3399$ &     $0.4418$   \\ \hline
 FWITIA        &   $0.2479$   &  $ 0.2417$\\ \hline
 FWHMOR        &   $0.2478$   & $0.2408$ \\ \hline
 A-FWBT       &  $0.2478$    & $0.2408$ \\ \hline
\end{tabular}
\end{table}
\begin{figure}[!h]
  \centering
  \includegraphics[width=6.5cm]{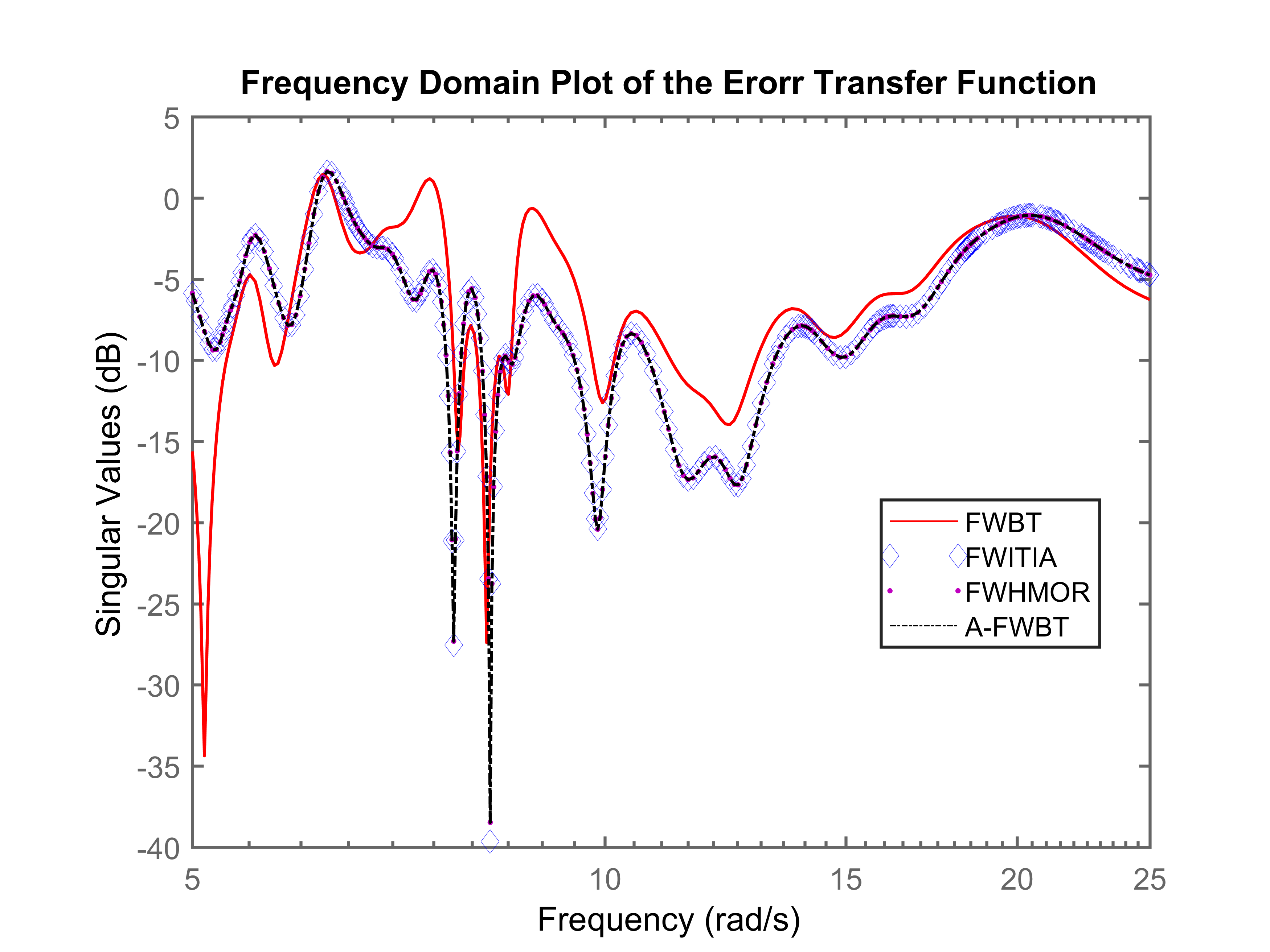}
  \caption{Singular values of $E(s)$ within $[5,25]$ rad/sec}\label{fig1}
\end{figure}
Further, ROMs of orders $6-15$ are obtained by using FWBT, FWITIA, FWHMOR, and A-FWBT. The weighted errors $||E_w(s)||_{\mathcal{H}_2}$ of the ROMs are compared in Figure \ref{fig2}, and it can be seen that FWHMOR and A-FWBT ensure high fidelity.
\begin{figure}[!h]
  \centering
  \includegraphics[width=6.5cm]{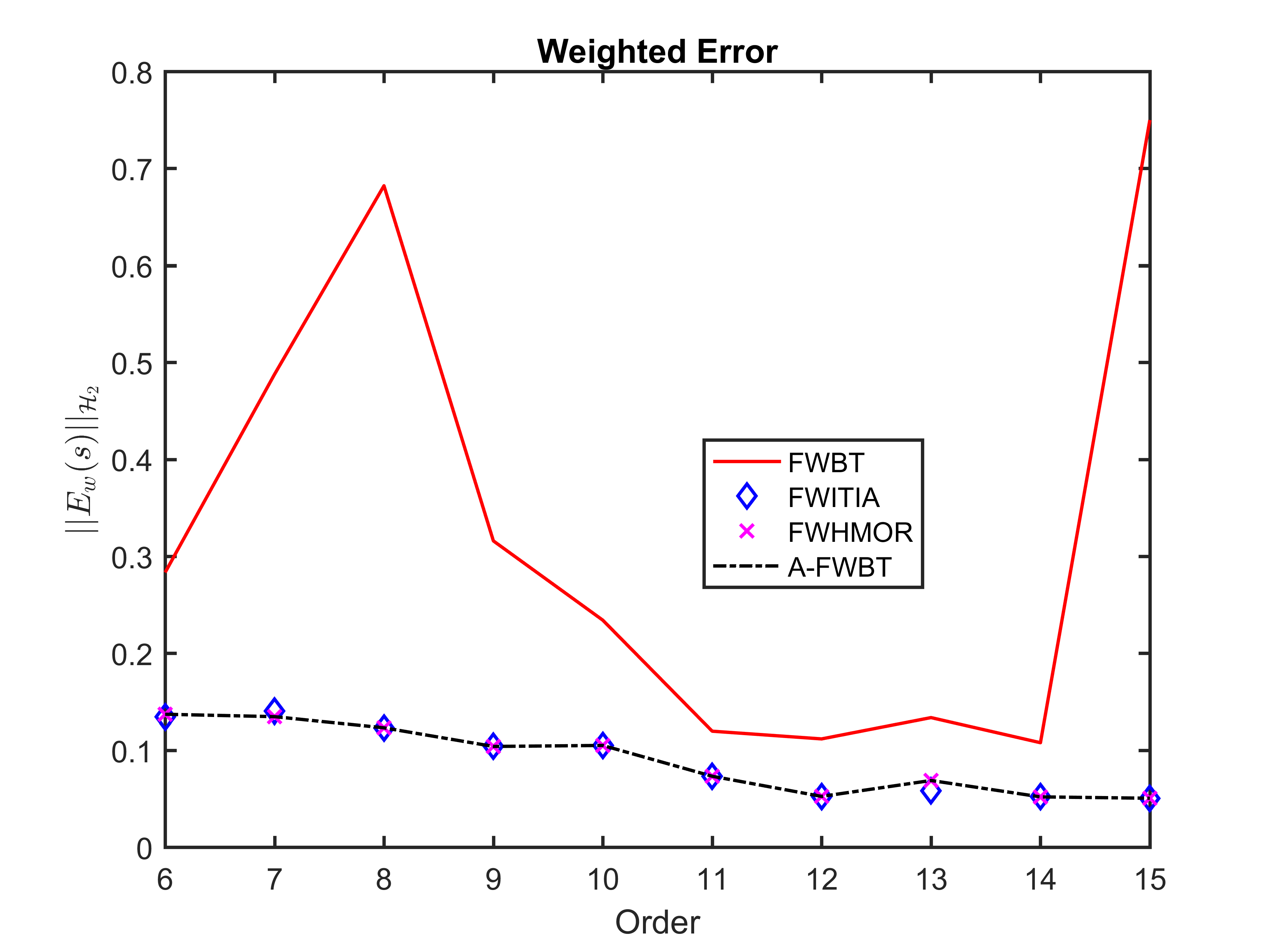}
  \caption{$\mathcal{H}_2$ norm of $E_w(s)$}\label{fig2}
\end{figure}
\subsection{Artificial Dynamic System} Consider $1006^{th}$ order artificial dynamic system model from the benchmark collection of \citep{chahlaoui2005benchmark}. Suppose a ROM of that artificial model is required, which ensures high fidelity within the frequency interval $[10,15]$ rad/sec. To achieve good accuracy within the desired frequency interval, a $4^{th}$ order band-pass filter with the passband $[10,15]$ rad/sec is used as the input and output weights, which is designed by using MATLAB's command \textit{butter(2,[10,15],'s')}. A $1^{st}$ order ROM is obtained by using FWBT, FWITIA, FWHMOR, and A-FWBT. The $\mathcal{H}_2$- and $\mathcal{H}_\infty$-norms of the weighted error transfer function $E_w(s)$ are tabulated in Table \ref{tab03}. It can be seen that FWHMOR and A-FWBT construct accurate ROMs. The singular values of $E(s)$ within $[10,15]$ rad/sec are plotted in Figure \ref{fig04}. It can be seen that FWHMOR and A-FWBT ensure good accuracy within the desired frequency region.
\begin{table}[!h]
\centering
\caption{Weighted Error}\label{tab03}
\begin{tabular}{|c|c|c|}
\hline
Technique &$||E_w(s)||_{\mathcal{H}_2}$&$||E_w(s)||_{\mathcal{H}_\infty}$\\ \hline
 FWBT         &$1.5736$ &     $1.4099$   \\ \hline
 FWITIA        &   $0.9335$   &  $ 0.8024$\\ \hline
 FWHMOR        &   $0.9334$   & $0.8028$ \\ \hline
 A-FWBT       &  $0.9334$    & $0.8028$ \\ \hline
\end{tabular}
\end{table}
\begin{figure}[!h]
  \centering
  \includegraphics[width=6.5cm]{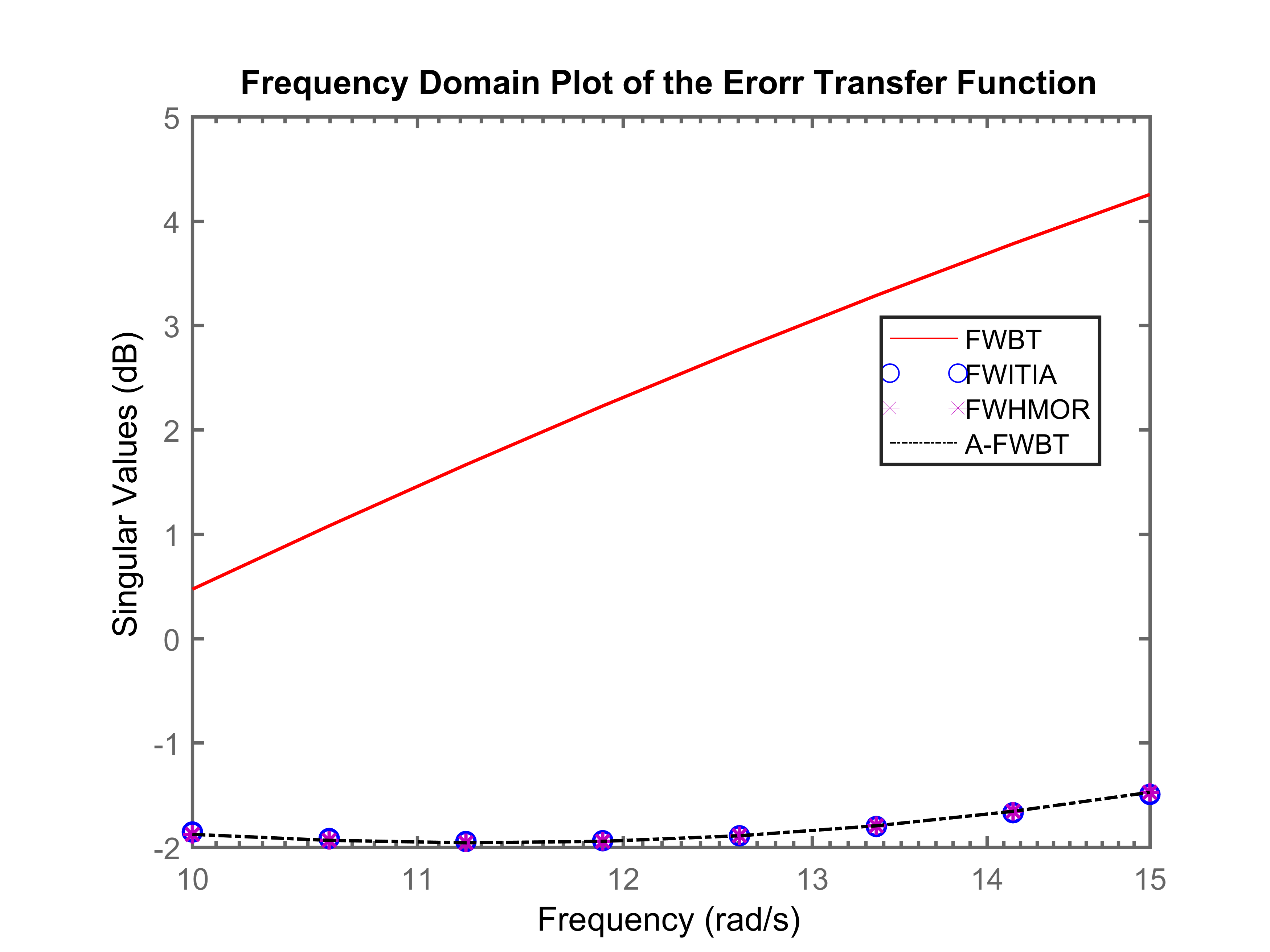}
  \caption{Singular values of $E(s)$ within $[10,15]$ rad/sec}\label{fig04}
\end{figure}
Further, ROMs of orders $5-15$ are obtained by using FWBT, FWITIA, FWHMOR, and A-FWBT. The weighted errors $||E_w(s)||_{\mathcal{H}_2}$ of the ROMs are compared in Figure \ref{fig3}, and it can be seen that FWHMOR and A-FWBT ensure high fidelity.
\begin{figure}[!h]
  \centering
  \includegraphics[width=6.5cm]{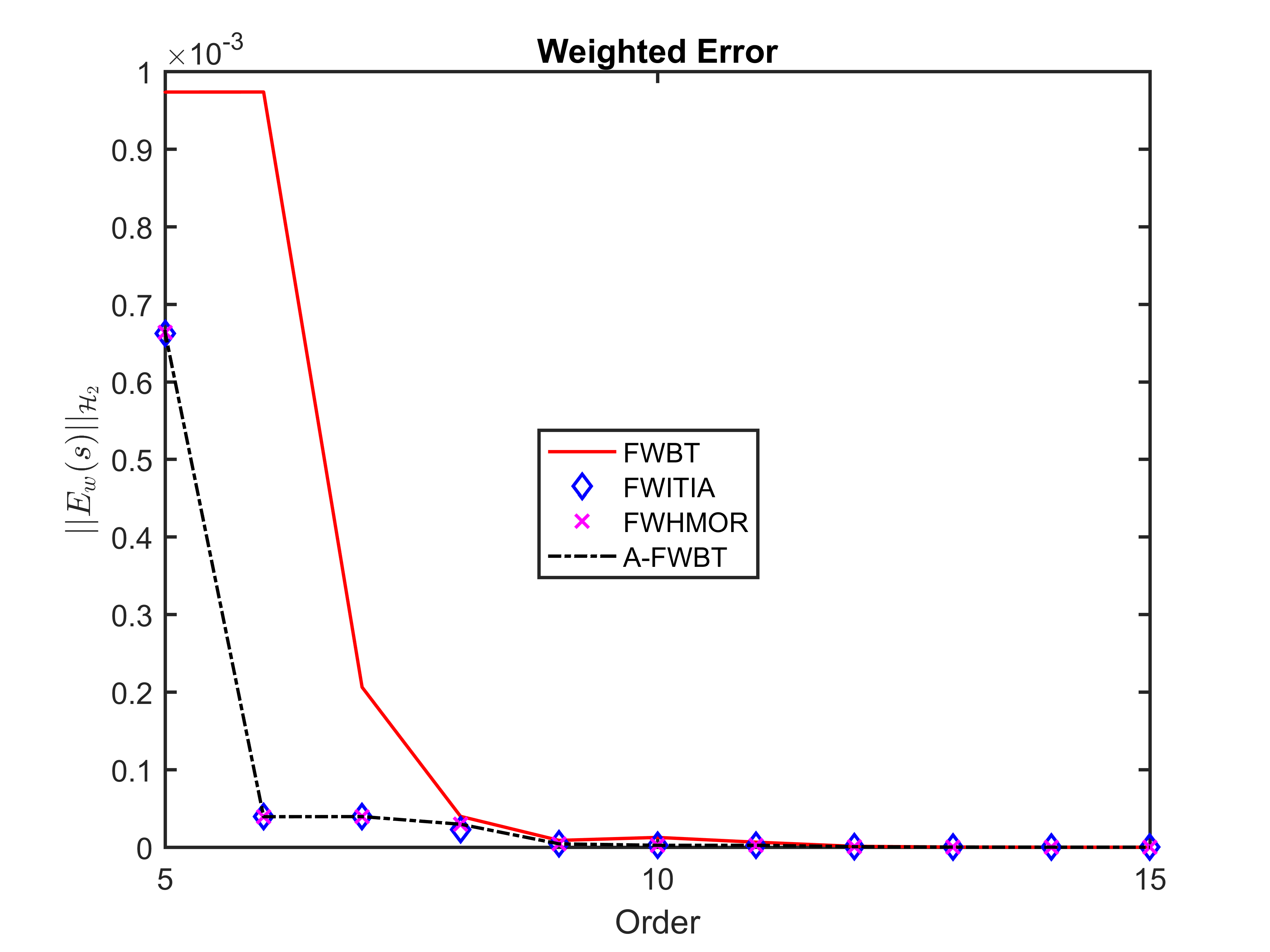}
  \caption{$\mathcal{H}_2$ norm of $E_w(s)$}\label{fig3}
\end{figure}
\subsection{International Space Station} Consider the $270^{th}$ order international space station model from the benchmark collection of \citep{chahlaoui2005benchmark} as the plant $P(s)$. An $\mathcal{H}_\infty$-controller $K(s)$ is designed using by MATLAB's \textit{ncfsyn} command wherein the loop shaping filter is specified as $\frac{20}{s+1.5}I_{3\times 3}$. The resulting controller is a $260^{th}$ order controller, which is reduced to $2^{nd}$ order controller $\tilde{K}(s)$ based on the closeness of the closed-loop transfer function criterion \citep{obinata2012model}. The frequency weights, which ensure that the closed-loop transfer function with the reduced controller is close to the original closed-loop transfer function, are given by, cf. \citep{obinata2012model},
\begin{align}
W_i(s)&=(I+P(s)K(s))^{-1}&\textnormal{ and }&&W_o(s)&=(I+P(s)K(s))^{-1}P(s).\nonumber
\end{align} The $\mathcal{H}_2$- and $\mathcal{H}_\infty$-norms of $E_w(s)=W_o(s)(K(s)-\tilde{K}(s))W_i(s)$ are tabulated in Table \ref{tab4}.
\begin{table}[!h]
\centering
\caption{Weighted Error}\label{tab4}
\begin{tabular}{|c|c|c|}
\hline
Technique &$||E_w(s)||_{\mathcal{H}_2}$&$||E_w(s)||_{\mathcal{H}_\infty}$\\ \hline
 FWBT         &$0.1361$ &     $4.7385$   \\ \hline
 FWITIA        &   $0.0066$   &  $0.0951$\\ \hline
 FWHMOR        &   $0.0066$   & $0.0950$ \\ \hline
 A-FWBT       &  $0.0066$    & $0.0950$ \\ \hline
\end{tabular}
\end{table} It can be noted that FWHMOR and A-FWBT show good accuracy. Further, ROMs of orders $3-15$ are obtained by using FWBT, FWITIA, FWHMOR, and A-FWBT. The weighted errors $||E_w(s)||_{\mathcal{H}_2}$ of the ROMs are compared in Figure \ref{fig4}, and it can be seen that FWHMOR and A-FWBT ensure high fidelity.
\begin{figure}[!h]
  \centering
  \includegraphics[width=6.5cm]{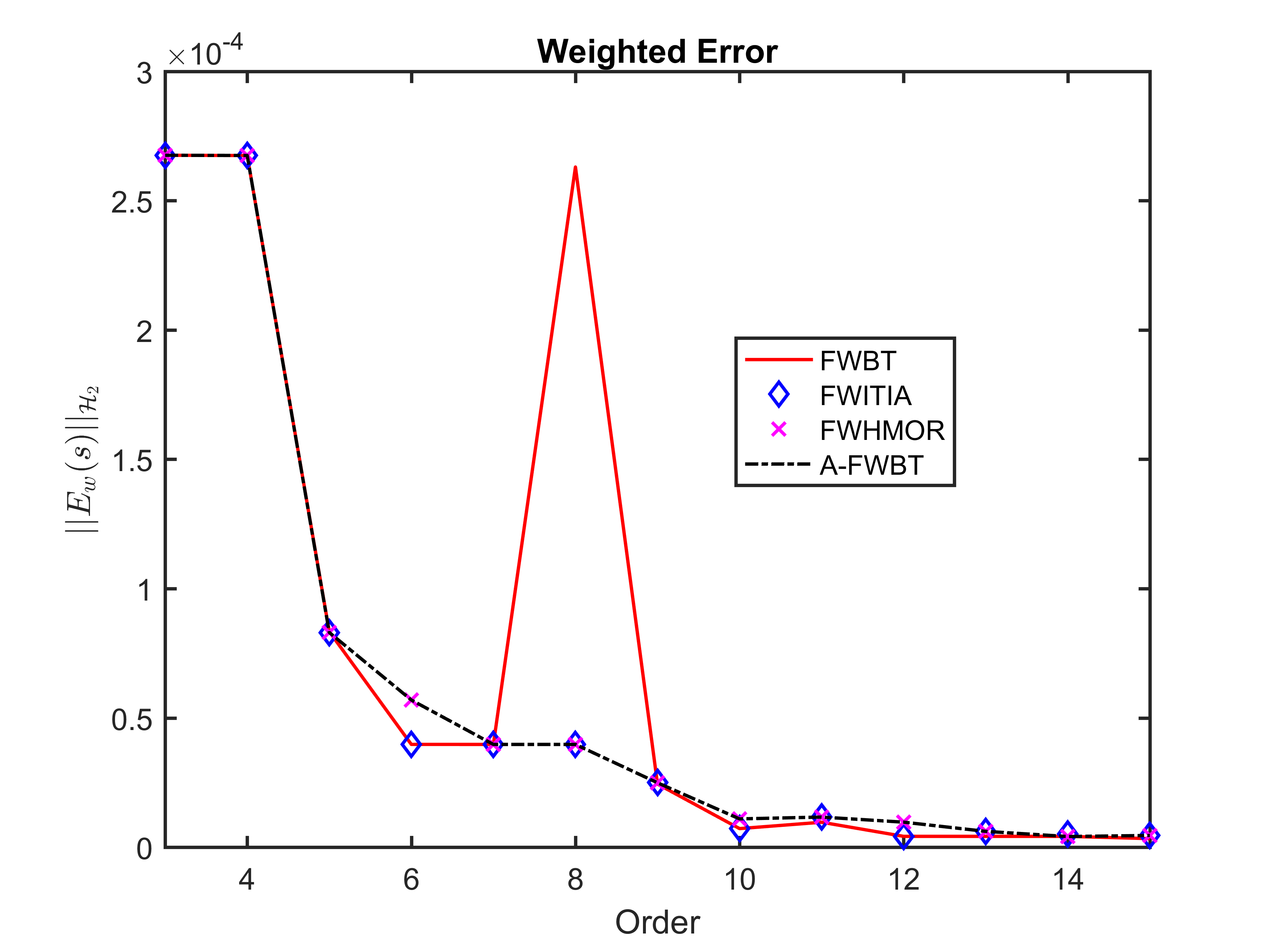}
  \caption{$\mathcal{H}_2$ norm of $E_w(s)$}\label{fig4}
\end{figure}
\section{Conclusion}\label{sec5}
We addressed the problem of frequency-weighted $\mathcal{H}_2$-optimal MOR within the projection framework. It is shown that although the first-order optimality conditions for the problem cannot be inherently met within the projection framework, the deviation in the optimality conditions decays as the order of the ROM increases. A fixed point iteration algorithm is proposed, which generates a nearly (local) optimal ROM. The oblique projection in the proposed algorithm is computed by solving \textit{sparse-dense} Sylvester equations for which several efficient algorithms exist. The numerical results validate the theory developed in the paper. In the future, a structure-preserving interpolation framework will be developed that preserves the structure of $F\big[H(s)\big]$ and $G\big[H(s)\big]$ in $F\big[\tilde{H}(s)\big]$ and $G\big[\tilde{H}(s)\big]$, respectively. This will enable satisfying the interpolation conditions (\ref{e13}) and (\ref{e14}) exactly, which is currently not possible in FWITIA.
\section*{Acknowledgment}
This work is supported in part by National Natural Science Foundation of China under Grant (No. $61873336$, $61873335$), in part by the National Key Research and Development Program (No. $2020$YFB $1708200$), in part by the Foreign Expert Program (No. $20$WZ$2501100$) granted by the Shanghai Science and Technology Commission of  Shanghai Municipality (Shanghai Administration of Foreign Experts Affairs), in part by $111$ Project (No. D$18003$) granted by the State Administration of Foreign Experts Affairs, and  in part by the Fundamental Research Funds for the Central Universities under Grant (No. FRF-BD-19-002A). M. I. Ahmad is supported by the Higher Education Commission of Pakistan under the National Research Program for Universities Project ID $10176$.

\begin{thebibliography}{}

\bibitem[\protect\astroncite{Ahmad et~al.}{2010a}]{ahmad2010krylov}
Ahmad, M.~I., Jaimoukha, I., and Frangos, M. (2010a).
\newblock Krylov subspace restart scheme for solving large-scale Sylvester
  equations.
\newblock In {\em Proceedings of the 2010 American Control Conference}, pages
  5726--5731. IEEE.

\bibitem[\protect\astroncite{Ahmad et~al.}{2010b}]{ahmad20100}
Ahmad, M.~I., Jaimoukha, I., and Frangos, M. (2010b).
\newblock $\mathcal{H}_2$ optimal model reduction of linear dynamical systems.
\newblock In {\em 49th IEEE Conference on Decision and Control (CDC)}, pages
  5368--5371. IEEE.

\bibitem[\protect\astroncite{Ani{\'c} et~al.}{2013}]{anic2013interpolatory}
Ani{\'c}, B., Beattie, C., Gugercin, S., and Antoulas, A.~C. (2013).
\newblock Interpolatory weighted-$\mathcal{H}_2$ model reduction.
\newblock {\em Automatica}, 49(5):1275--1280.

\bibitem[\protect\astroncite{Antoulas}{2005}]{antoulas2005approximation}
Antoulas, A.~C. (2005).
\newblock {\em Approximation of large-scale dynamical systems}.
\newblock SIAM.

\bibitem[\protect\astroncite{Beattie and Gugercin}{2014}]{beattie2014model}
Beattie, C.~A. and Gugercin, S. (2014).
\newblock Model reduction by rational interpolation.
\newblock {\em Model Reduction and Algorithms: Theory and Applications, P.
  Benner, A. Cohen, M. Ohlberger, and K. Willcox, eds., Comput. Sci. Engrg},
  15:297--334.

\bibitem[\protect\astroncite{Benner et~al.}{2017}]{benner2017model}
Benner, P., Cohen, A., Ohlberger, M., and Willcox, K. (2017).
\newblock {\em Model reduction and approximation: theory and algorithms},
  volume~15.
\newblock SIAM.

\bibitem[\protect\astroncite{Benner et~al.}{2011}]{bennersparse}
Benner, P., K{\"o}hler, M., and Saak, J. (2011).
\newblock Sparse-dense Sylvester equations in $\mathcal{H}_2$-model order reduction. MPI
  Magdeburg preprints MPIMD/11-11, 2011.

\bibitem[\protect\astroncite{Benner and
  K{\"u}rschner}{2014}]{benner2014computing}
Benner, P. and K{\"u}rschner, P. (2014).
\newblock Computing real low-rank solutions of Sylvester equations by the
  factored ADI method.
\newblock {\em Computers \& Mathematics with Applications}, 67(9):1656--1672.

\bibitem[\protect\astroncite{Benner et~al.}{2016}]{benner2016frequency}
Benner, P., K{\"u}rschner, P., and Saak, J. (2016).
\newblock Frequency-limited balanced truncation with low-rank approximations.
\newblock {\em SIAM Journal on Scientific Computing}, 38(1):A471--A499.

\bibitem[\protect\astroncite{Benner et~al.}{2005}]{benner2005dimension}
Benner, P., Mehrmann, V., and Sorensen, D.~C. (2005).
\newblock {\em Dimension reduction of large-scale systems}, volume~45.
\newblock Springer.

\bibitem[\protect\astroncite{Breiten et~al.}{2015}]{breiten2015near}
Breiten, T., Beattie, C., and Gugercin, S. (2015).
\newblock Near-optimal frequency-weighted interpolatory model reduction.
\newblock {\em Systems \& Control Letters}, 78:8--18.

\bibitem[\protect\astroncite{Castagnotto et~al.}{2017}]{castagnotto2017sss}
Castagnotto, A., Varona, M.~C., Jeschek, L., and Lohmann, B. (2017).
\newblock SSS \& SSSMOR: Analysis and reduction of large-scale dynamic systems
  in MATLAB.
\newblock {\em at-Automatisierungstechnik}, 65(2):134--150.

\bibitem[\protect\astroncite{Chahlaoui and
  Van~Dooren}{2005}]{chahlaoui2005benchmark}
Chahlaoui, Y. and Van~Dooren, P. (2005).
\newblock Benchmark examples for model reduction of linear time-invariant
  dynamical systems.
\newblock In {\em Dimension Reduction of Large-Scale Systems}, pages 379--392.
  Springer.

\bibitem[\protect\astroncite{Davis}{2004}]{davis2004algorithm}
Davis, T.~A. (2004).
\newblock Algorithm 832: UMFPACK v4. 3---an unsymmetric-pattern multifrontal
  method.
\newblock {\em ACM Transactions on Mathematical Software (TOMS)},
  30(2):196--199.

\bibitem[\protect\astroncite{Demmel et~al.}{1999a}]{demmel1999supernodal}
Demmel, J.~W., Eisenstat, S.~C., Gilbert, J.~R., Li, X.~S., and Liu, J.~W.
  (1999a).
\newblock A supernodal approach to sparse partial pivoting.
\newblock {\em SIAM Journal on Matrix Analysis and Applications},
  20(3):720--755.

\bibitem[\protect\astroncite{Demmel et~al.}{1999b}]{demmel1999asynchronous}
Demmel, J.~W., Gilbert, J.~R., and Li, X.~S. (1999b).
\newblock An asynchronous parallel supernodal algorithm for sparse Gaussian
  elimination.
\newblock {\em SIAM Journal on Matrix Analysis and Applications},
  20(4):915--952.

\bibitem[\protect\astroncite{Diab et~al.}{2000}]{diab2000optimal}
Diab, M., Liu, W., and Sreeram, V. (2000).
\newblock Optimal model reduction with a frequency weighted extension.
\newblock {\em Dynamics and Control}, 10(3):255--276.

\bibitem[\protect\astroncite{Enns}{1984}]{enns1984model}
Enns, D.~F. (1984).
\newblock Model reduction with balanced realizations: An error bound and a
  frequency weighted generalization.
\newblock In {\em The 23rd IEEE conference on decision and control}, pages
  127--132. IEEE.

\bibitem[\protect\astroncite{Ghafoor et~al.}{2007}]{ghafoor2007frequency}
Ghafoor, A., Sreeram, V., and Treasure, R. (2007).
\newblock Frequency weighted model reduction technique retaining Hankel
  singular values.
\newblock {\em Asian Journal of Control}, 9(1):50--56.

\bibitem[\protect\astroncite{Ghafoor and Sreeram}{2008}]{ghafoor2008survey}
Ghafoor, A. and Sreeram, V. (2008).
\newblock A survey/review of frequency-weighted balanced model reduction
  techniques.
\newblock {\em Journal of Dynamic Systems, Measurement, and Control}, 130(6).

\bibitem[\protect\astroncite{Gugercin et~al.}{2008}]{gugercin2008h_2}
Gugercin, S., Antoulas, A.~C., and Beattie, C. (2008).
\newblock $\mathcal{H}_2$ model reduction for large-scale linear dynamical
  systems.
\newblock {\em SIAM journal on matrix analysis and applications},
  30(2):609--638.

\bibitem[\protect\astroncite{Gugercin et~al.}{2003}]{gugercin2003modified}
Gugercin, S., Sorensen, D.~C., and Antoulas, A.~C. (2003).
\newblock A modified low-rank Smith method for large-scale Lyapunov equations.
\newblock {\em Numerical Algorithms}, 32(1):27--55.

\bibitem[\protect\astroncite{Halevi}{1990}]{halevi1990frequency}
Halevi, Y. (1990).
\newblock Frequency weighted model reduction via optimal projection.
\newblock In {\em 29th IEEE Conference on Decision and Control}, pages
  2906--2911. IEEE.

\bibitem[\protect\astroncite{Huang et~al.}{2001}]{huang2001new}
Huang, X.-X., Yan, W.-Y., and Teo, K. (2001).
\newblock A new approach to frequency weighted $L_2$ optimal model reduction.
\newblock {\em International Journal of Control}, 74(12):1239--1246.

\bibitem[\protect\astroncite{Hurak et~al.}{2001}]{hurak2001discussion}
Hurak, Z., Sreeram, V., Wang, G., Van~Gestel, T., De~Moor, B., Anderson, B.,
  and Van~Overschee, P. (2001).
\newblock Discussion on ``On frequency weighted balanced truncation: Hankel
  singular values and error bounds" by T. Van gestel, B. De Moor, BDO Anderson,
  and P. Van Overschee.
\newblock {\em European Journal of Control}, 7(6):593--595.

\bibitem[\protect\astroncite{Ibrir}{2018}]{ibrir2018projection}
Ibrir, S. (2018).
\newblock A projection-based algorithm for model-order reduction with
  $\mathcal{H}_2$ performance: A convex-optimization setting.
\newblock {\em Automatica}, 93:510--519.

\bibitem[\protect\astroncite{K{\"u}rschner}{2018}]{kurschner2018balanced}
K{\"u}rschner, P. (2018).
\newblock Balanced truncation model order reduction in limited time intervals
  for large systems.
\newblock {\em Advances in Computational Mathematics}, 44(6):1821--1844.

\bibitem[\protect\astroncite{Li and White}{2002}]{li2002low}
Li, J.-R. and White, J. (2002).
\newblock Low rank solution of Lyapunov equations.
\newblock {\em SIAM Journal on Matrix Analysis and Applications},
  24(1):260--280.

\bibitem[\protect\astroncite{Li et~al.}{1999}]{li1999h}
Li, L., Xie, L., Yan, W.-Y., and Soh, Y.~C. (1999).
\newblock $\mathcal{H}_2$ optimal reduced-order filtering with frequency weighting.
\newblock {\em IEEE Transactions on Circuits and Systems I: Fundamental Theory
  and Applications}, 46(6):763--767.

\bibitem[\protect\astroncite{Moore}{1981}]{moore1981principal}
Moore, B. (1981).
\newblock Principal component analysis in linear systems: Controllability,
  observability, and model reduction.
\newblock {\em IEEE transactions on automatic control}, 26(1):17--32.

\bibitem[\protect\astroncite{Obinata and Anderson}{2012}]{obinata2012model}
Obinata, G. and Anderson, B.~D. (2012).
\newblock {\em Model reduction for control system design}.
\newblock Springer Science \& Business Media.

\bibitem[\protect\astroncite{Panzer}{2014}]{panzer2014model}
Panzer, H.~K. (2014).
\newblock {\em Model order reduction by Krylov subspace methods with global
  error bounds and automatic choice of parameters}.
\newblock PhD thesis, Technische Universit{\"a}t M{\"u}nchen.

\bibitem[\protect\astroncite{Penzl}{1999a}]{penzl1999cyclic}
Penzl, T. (1999a).
\newblock A cyclic low-rank Smith method for large sparse Lyapunov equations.
\newblock {\em SIAM Journal on Scientific Computing}, 21(4):1401--1418.

\bibitem[\protect\astroncite{Penzl}{1999b}]{penzl1999lyapack}
Penzl, T. (1999b).
\newblock Lyapack: A MATLAB toolbox for large Lyapunov and Riccati equations.
\newblock {\em Model Reduction Problems, and Linear-Quadratic Optimal Control
  Problems, SFB}, 393.

\bibitem[\protect\astroncite{Petersson}{2013}]{petersson2013nonlinear}
Petersson, D. (2013).
\newblock {\em A nonlinear optimization approach to $\mathcal{H}_2$-optimal modeling and
  control}.
\newblock PhD thesis, Link{\"o}ping University.

\bibitem[\protect\astroncite{Rommes and Martins}{2006}]{rommes2006efficient}
Rommes, J. and Martins, N. (2006).
\newblock Efficient computation of multivariable transfer function dominant
  poles using subspace acceleration.
\newblock {\em IEEE transactions on power systems}, 21(4):1471--1483.

\bibitem[\protect\astroncite{Saak et~al.}{2010}]{saak2010matrix}
Saak, J., Mena, H., and Benner, P. (2010).
\newblock Matrix equation sparse solvers (MESS): a MATLAB toolbox for the
  solution of sparse large-scale matrix equations.
\newblock {\em Chemnitz University of Technology, Germany}.

\bibitem[\protect\astroncite{Sahlan et~al.}{2007}]{sahlan2007properties}
Sahlan, S., Ghafoor, A., and Sreeram, V. (2007).
\newblock Properties of frequency weighted balanced truncation techniques.
\newblock In {\em 2007 IEEE International Conference on Automation Science and
  Engineering}, pages 765--770. IEEE.

\bibitem[\protect\astroncite{Spanos et~al.}{1990}]{spanos1990optimal}
Spanos, J., Milman, M., and Mingori, D. (1990).
\newblock Optimal model reduction and frequency-weighted extension.
\newblock In {\em Guidance, Navigation and Control Conference}, page 3345.

\bibitem[\protect\astroncite{Sreeram}{2002}]{sreeram2002properties}
Sreeram, V. (2002).
\newblock On the properties of frequency weighted balanced truncation
  techniques.
\newblock In {\em Proceedings of the 2002 American Control Conference (IEEE
  Cat. No. CH37301)}, volume~3, pages 1753--1754. IEEE.

\bibitem[\protect\astroncite{Sreeram and Sahlan}{2012}]{sreeram2012improved}
Sreeram, V. and Sahlan, S. (2012).
\newblock Improved results on frequency-weighted balanced truncation and error
  bounds.
\newblock {\em International Journal of Robust and Nonlinear Control},
  22(11):1195--1211.

\bibitem[\protect\astroncite{Van~Dooren et~al.}{2008}]{van2008h2}
Van~Dooren, P., Gallivan, K.~A., and Absil, P.-A. (2008).
\newblock $\mathcal{H}_2$-optimal model reduction of MIMO systems.
\newblock {\em Applied Mathematics Letters}, 21(12):1267--1273.

\bibitem[\protect\astroncite{Wang et~al.}{1999}]{wang1999new}
Wang, G., Sreeram, V., and Liu, W. (1999).
\newblock A new frequency-weighted balanced truncation method and an error
  bound.
\newblock {\em IEEE Transactions on Automatic Control}, 44(9):1734--1737.

\bibitem[\protect\astroncite{Wolf}{2014}]{wolf2014h}
Wolf, T. (2014).
\newblock {\em $\mathcal{H}_2$ pseudo-optimal model order reduction}.
\newblock PhD thesis, Technische Universit{\"a}t M{\"u}nchen.

\bibitem[\protect\astroncite{Yan and Lam}{1999}]{yan1999approximate}
Yan, W.-Y. and Lam, J. (1999).
\newblock An approximate approach to $\mathcal{H}_2$ optimal model reduction.
\newblock {\em IEEE Transactions on Automatic Control}, 44(7):1341--1358.

\bibitem[\protect\astroncite{Yan et~al.}{1997}]{yan1997convergent}
Yan, W.-Y., Xie, L., and Lam, J. (1997).
\newblock Convergent algorithms for frequency weighted $L_2$ model reduction.
\newblock {\em Systems \& control letters}, 31(1):11--20.

\bibitem[\protect\astroncite{Zulfiqar and Sreeram}{2018}]{zulfiqar2018weighted}
Zulfiqar, U. and Sreeram, V. (2018).
\newblock Weighted iterative tangential interpolation algorithms.
\newblock In {\em 2018 Australian \& New Zealand Control Conference (ANZCC)},
  pages 380--384. IEEE.

\bibitem[\protect\astroncite{Zulfiqar et~al.}{2021}]{zulfiqar2019frequency}
Zulfiqar, U., Sreeram, V., Ahmad, M.~I., and Du, X. (2021).
\newblock Frequency-weighted $\mathcal{H}_2$-pseudo-optimal model order reduction.
\newblock {\em IMA Journal of Mathematical Control and Information}.

\bibitem[\protect\astroncite{Zulfiqar et~al.}{2017}]{zulfiqar2017passivity}
Zulfiqar, U., Tariq, W., Li, L., and Liaquat, M. (2017).
\newblock A passivity-preserving frequency-weighted model order reduction
  technique.
\newblock {\em IEEE Transactions on Circuits and Systems II: Express Briefs},
  64(11):1327--1331.

\end{thebibliography}

\end{document}